\documentclass[11pt,a4paper]{article}
\def\Label#1{}
\newcommand{\pdfextension}{pdf}
\newcommand{\pngextension}{png}
\pdfoutput=1 

\let\rho=\varrho

\usepackage{ifthen}
\def\defcommand#1:#2?{\global\let\myref\myundefined%
	\ifthenelse{\equal{#1}{e}}{\global\let\myref\eref}{}%
	\ifthenelse{\equal{#1}{p}}{\global\let\myref\pref}{}%
	\ifthenelse{\equal{#1}{l}}{\global\let\myref\lref}{}%
	\ifthenelse{\equal{#1}{r}}{\global\let\myref\rref}{}%
	\ifthenelse{\equal{#1}{t}}{\global\let\myref\tref}{}%
	\ifthenelse{\equal{#1}{f}}{\global\let\myref\fref}{}%
	\ifthenelse{\equal{#1}{s}}{\global\let\myref\sref}{}%
	\ifthenelse{\equal{#1}{d}}{\global\let\myref\dref}{}%
}
\let\ssref=\ref
\def\fref#1{Figure~\ssref{#1}}

\def\cref#1{Condition~\ssref{#1}}
\def\Cref#1{Corollary~\ssref{#1}}
\def\eref#1{(\ssref{#1})}
\def\sref#1{\textsection\ssref{#1}}
\def\lref#1{Lemma~\ssref{#1}}
\def\rref#1{Remark~\ssref{#1}}
\def\tref#1{Theorem~\ssref{#1}}
\def\dref#1{Definition~\ssref{#1}}
\def\pref#1{Proposition~\ssref{#1}}
\def\aref#1{Assumption~\ssref{#1}}

\def\dref#1{Definition~\ssref{#1}}
\def\myundefined#1{
{what is #1?}%
}
\def\ref#1{\defcommand#1?%
	\myref{#1}}
\let\myref\relax
\newenvironment{myitem}
{\begin{itemize}
  \setlength{\itemsep}{1pt}
  \setlength{\parskip}{0pt}
  \setlength{\parsep}{0pt}}
{\end{itemize}}

\newenvironment{myenum}
{\begin{enumerate}
  \setlength{\itemsep}{1pt}
  \setlength{\parskip}{0pt}
  \setlength{\parsep}{0pt}}
{\end{enumerate}}

\usepackage{amsmath}
\usepackage{amsfonts}
\usepackage{graphicx}
\usepackage{times}
\usepackage{amsthm}
\usepackage{ amssymb }
\usepackage{color}
\usepackage{mhequ}
\usepackage{dsfont}
\usepackage[margin=2.9cm]{geometry}
\usepackage{color}
\usepackage{url}
\usepackage{lipsum}
\usepackage{authblk}
\usepackage{subfig}
\usepackage{mhequ}

\usepackage{tikz}
\usepackage[graphics, active, tightpage]{}

\usepackage[expansion=true]{microtype}

\def\thecomma{\ifx,\thenewxt \else\ifx;\thenext \else\ifx.\thenext
	\else\ifx!\thenext \else\ifx:\thenext\else\ifx)\thenext \else \
	\fi\fi\fi\fi\fi\fi}
\def\condblank{\futurelet\thenext\thecomma}
\def\ie{{\it i.e.,}\condblank}
\def\eg{{\it e.g.,}\condblank}

\numberwithin{equation}{section}

\newtheorem{theorem}{Theorem}[section]
\newtheorem{lemma}[theorem]{Lemma}
\newtheorem{proposition}[theorem]{Proposition}
\newtheorem{definition}[theorem]{Definition}

\newtheorem{assumption}[theorem]{Assumption}

\theoremstyle{definition} 
\newtheorem{remark}[theorem]{Remark}

\usepackage{cite}

\usepackage{stmaryrd}

\newcommand{\dd}{\mathrm{d}}
\newcommand{\dt}{\,\dd t}
\newcommand{\mpart}[2]{\frac{\partial #1}{\partial #2}}
\newcommand{\avg}[1]{\left\langle #1\right\rangle}

\newcommand{\bigoh}[1]{\hat {\mathcal O}(p_2^{#1})}
\newcommand{\bigohneg}[1]{\hat {\mathcal O}\!\left({p_2^{-#1}}\right)}

\def\argcdot{{\,\cdot\,}}

\newcommand{\cU}{{\ensuremath{\mathcal U}} }

\newcommand{\bbN}{{\ensuremath{\mathbb N}} }

\newcommand{\bbR}{{\ensuremath{\mathbb R}} }

\newcommand{\bbT}{{\ensuremath{\mathbb T}} }

\newcommand{\bbZ}{{\ensuremath{\mathbb Z}} }

\newcommand{\ga}{\alpha}

\newcommand{\gd}{\delta}

\newcommand{\gO}{\Omega}

\let\kappa=\varkappa
\let\phi=\varphi

\def\KK{{\mathcal K}}

\newcommand{\ind}{\mathbf{1}}
\def\p2t2{{\tilde p_2^{\,2}}}

\begin{document}

\title{Non-equilibrium steady state and subgeometric ergodicity for a chain of
three coupled rotors}
\author[1]{N. Cuneo}
\author[1,2]{J.-P. Eckmann}
\author[3]{C. Poquet}
\affil[1]{D\'epartement de physique th\'eorique, Universit\'e de Gen\`eve}
\affil[2]{Section de math\'ematiques, Universit\'e de Gen\`eve}
\affil[3]{Dipartimento di Matematica, Universit\`a di Roma Tor Vergata}
\date{} 

\maketitle
\thispagestyle{empty}
\begin{abstract}
We consider a chain of three rotors (rotators) whose ends are coupled to
stochastic heat baths. The temperatures of the two baths can be different, and
we  allow some constant torque to be applied at each end of the
chain. Under some non-degeneracy condition on the interaction potentials, we
show that the
process admits a unique invariant probability measure, and that it is ergodic
with a stretched exponential rate. 
The interesting issue is to estimate the
rate at which 
the energy of the middle rotor decreases. As it is not
directly connected to the heat baths, its energy can only be
dissipated through the two outer rotors. But when the middle rotor spins
very rapidly, it fails to interact effectively with its neighbours due to the
rapid oscillations of the forces. By averaging techniques, we obtain an
effective dynamics for the middle rotor, which then enables us to find a
 Lyapunov function. This and an irreducibility
argument give the 
desired result. We finally illustrate numerically some properties of the
non-equilibrium steady state.
\end{abstract}

\section{Introduction}

Hamiltonian chains of mechanical
oscillators have been studied for a long time. 
Several models describe a linear chain of masses, with polynomial
\emph{interaction} potentials 
between adjacent masses, and \emph{pinning} potentials which tie the masses
down in the laboratory frame. Under the assumption that the interaction
is stronger than the pinning, it was shown in
\cite{eckmann_nonequilibrium_1999} that the model has
an invariant probability measure
when the chain is attached at each extremity to two heat baths at different
temperatures. That paper, and later developments, see \eg
\cite{eckmann_hairer_2000},  relied on analytic arguments, showing in particular
that the infinitesimal generator has compact resolvent in a suitable
function space. 

Two elements were added later in the paper
\cite{reybellet_exponential_2002}: First, the authors used a more
probabilistic approach, based on Harris recurrence  
as developed by Meyn and Tweedie \cite{meyn_markov_2009}. Second, a
detailed analysis allowed them
to understand the transfer of energy
from the central oscillators to the (dissipative) baths. In that case
the convergence to 
the stationary state is of exponential rate. In \cite{carmona_2007},
this reasoning was extended to more general contexts.

The dynamics of the chain is very different when the pinning potential is  
\emph{stronger} than the interaction potential. In that case the
chain may have breathers, \ie oscillators concentrating a lot of energy, 
which is transferred only very slowly to their neighbours. This may
lead to subexponential ergodicity, as shown by 
Hairer and Mattingly \cite{hairer_slow_2009}
in the case of a chain of 3 oscillators with strong pinning.

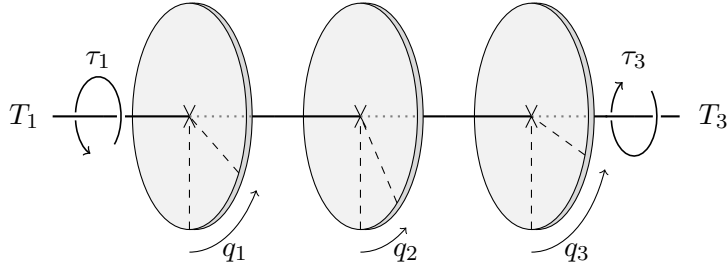
\begin{figure}[ht]
\centering

\begin{tikzpicture}[scale=1.5]
\tikzset{
    partial ellipse/.style args={#1:#2:#3}{
        insert path={+ (#1:#3) arc (#1:#2:#3)}
    }
}

   \begin{scope}[shift={(-1.5,0)}]
        \filldraw[fill=gray!30] (0.05,0.0) ellipse (0.5 and 1);
		\fill[gray!30] (0, 1.0) rectangle (0.05,-1);
		\draw (0,1.0)--(0.05,1.0);
		\draw (0,-1.0)--(0.05,-1.0);
        \filldraw[fill=gray!10] (0.0,0) ellipse (0.5 and 1);
		
		\draw[->] (0,0) [partial ellipse=-90:-33:0.7 and 1.2];
		
		\draw (-0.05, -0.1)--(0.05,0.1);
		\draw (0.05, -0.1)--(-0.05,0.1);
		\draw[thick] (-0.75, 0)--(0,0);
		\draw[thick, dotted, gray] (0, 0)--(0.55,0);
		\draw[thick] (0.55, 0)--(0.75,0);
		\draw[dashed] (0,-1)--(0,0);
		\draw[dashed] (0.43301270189,-0.5)--(0,0);
		\node[] at (0.4,-1.2) {$q_1$};
	\end{scope};

	\begin{scope}[shift={(0,0)}]
        \filldraw[fill=gray!30] (0.05,0.0) ellipse (0.5 and 1);
		\fill[gray!30] (0, 1.0) rectangle (0.05,-1);
		\draw (0,1.0)--(0.05,1.0);
		\draw (0,-1.0)--(0.05,-1.0);
        \filldraw[fill=gray!10] (0.0,0) ellipse (0.5 and 1);
		
		\draw[->] (0,0) [partial ellipse=-90:-55:0.7 and 1.2];
		
		\draw (-0.05, -0.1)--(0.05,0.1);
		\draw (0.05, -0.1)--(-0.05,0.1);
		\draw[thick] (-0.75, 0)--(0,0);
		\draw[thick, dotted, gray] (0, 0)--(0.5,0);
		\draw[thick] (0.55, 0)--(0.75,0);
		\draw[dashed] (0,-1)--(0,0);
		\draw[dashed] (0.32139380484,-0.76604444311)--(0,0);
		\node[] at (0.4,-1.2) {$q_2$};
	\end{scope};

   \begin{scope}[shift={(1.5,0)}]
        \filldraw[fill=gray!30] (0.05,0.0) ellipse (0.5 and 1);
		\fill[gray!30] (0, 1.0) rectangle (0.05,-1);
		\draw (0,1.0)--(0.05,1.0);
		\draw (0,-1.0)--(0.05,-1.0);
        \filldraw[fill=gray!10] (0.0,0) ellipse (0.5 and 1);
		
		\draw[->] (0,0) [partial ellipse=-90:-23:0.7 and 1.2];
		
		\draw (-0.05, -0.1)--(0.05,0.1);
		\draw (0.05, -0.1)--(-0.05,0.1);
		\draw[thick] (-0.75, 0)--(0,0);
		\draw[thick, dotted, gray] (0, 0)--(0.5,0);
		\draw[thick] (0.55, 0)--(0.75,0);
		\draw[dashed] (0,-1)--(0,0);
		\draw[dashed] (0.46984631039,-0.34202014332)--(0,0);
		\node[] at (0.4,-1.2) {$q_3$};
	\end{scope};

	\node[] at (-2.95,0) {$T_1$};
	\node[] at (3.1,0) {$T_3$};

	\draw[line width=3pt, white] (-2.3,0) [partial ellipse=-45:250:0.2 and 0.35];
	\draw[->, line width=0.6pt] (-2.3,0) [partial ellipse=-45:250:0.2 and 0.35];
	\draw[line width=3pt, white] (-2.4, 0)--(-2.7,0);
	\draw[thick] (-2.25, 0)--(-2.7,0);
	\node[] at (-2.3,0.5) {$\tau_1$};

	\draw[thick] (2.4, 0)--(2.8,0);
	\draw[line width=3pt, white] (2.4,0) [partial ellipse=40:-240:0.2 and 0.35];
	\draw[->, line width=0.6pt] (2.4,0) [partial ellipse=40:-240:0.2 and 0.35];
	\draw[line width=3pt, white] (2.17, 0)--(2.3,0);
	\draw[thick] (2.15, 0)--(2.4,0);
	\node[] at (2.4,0.5) {$\tau_3$};

\end{tikzpicture}
\caption{A chain of three rotors with two external torques $\tau_1$ and $\tau_3$
and two heat baths at temperatures $T_1$ and $T_3$.}
\label{f:system}
\end{figure}

In this paper, we discuss a model with  
three rotors (see \fref{f:system}), each given by an angle $q_i\in \bbT =
\bbR/2\pi
\bbZ$ and a momentum $p_i \in \bbR$, $i=1,2,3$. The phase space is therefore
$\Omega = \bbT^3 \times \mathbb R^3$, and we will consider the measure space
$(\Omega, \mathcal B)$,
where $\mathcal B$ is the Borel $\sigma$-field over $\Omega$. We will denote the
points
of $\Omega$ by $x = (q, p)$ with $q = (q_1, q_2, q_3)$ and $p=(p_1, p_2, p_3)$.

We introduce the Hamiltonian
\begin{equ}
H(q, p) = \sum_{i=1}^3 \left({\frac 1 2} p_i^2 +
U_i(q_i)\right) + \sum_{b=1,3}W_b(q_2-q_b)~,
\end{equ}
with some smooth {\em interaction potentials} $W_b:\bbT \to \bbR$, $b=1,3$, and
some smooth {\em pinning potentials} $U_i:\bbT \to \bbR, i=1,2,3$. We now let
the two outer rotors (\ie the rotors 1 and 3) interact with
Langevin-type heat baths at temperatures $T_1, T_3 >0$, and with coupling
constants
$\gamma_1, \gamma_3 > 0$. Moreover, we apply some constant (possibly zero)
external
forces $\tau_1$ and $\tau_3$ to the two outer rotors.
Introducing $w_b = W_b'$
and $u_i = U_i'$, we obtain the system of SDE:
\begin{equs}[eq:sde]
	\dd q_i(t)&=p_i(t)\dt~, \qquad i=1,2,3~,\\
	\dd p_2(t)&=-\sum_{b=1,3}w_b\big(q_2(t)-q_b(t)\big) \dt
-u_2(q_2(t))\dt~,\\
	\dd p_b(t)&=\Big(w_b\big(q_2(t)-q_b(t)\big)
+\tau_b-u_b\big(q_b(t)\big)-\gamma_b
	p_b(t)\Big)\dt +\sqrt{2\gamma_b T_b}\,\dd B^b_t~,\qquad b=1,3~,
\end{equs}
where $B^1$ and $B^3$ are standard independent Brownian motions.

\medskip\noindent{\bf Notation}. In the sequel, the index $b$ always refers to
the rotors 1 and 3 at the boundaries of the chain, and we write $\sum_b$
instead of $\sum_{b=1,3}$.

\begin{remark}
Our model can be viewed as an extreme case of that studied in
\cite{hairer_slow_2009}.
A key factor in that paper is to realise how the 
frequency of one isolated pinned oscillator depends on its
energy. Indeed, for an
isolated oscillator with  Hamiltonian $p^2/2 + q^{2k}/(2k)$, the
frequency grows like the energy to the power $\frac 12 - \frac 1{2k}$. When
$k\to\infty$, the exponent converges to $\frac 12$. In this limit, the
pinning potential formally becomes an infinite potential well, so that the
variable $q$ is constrained to a compact interval. In our model, the
position (angle) of a rotor lives in a compact space, and its frequency
scales like its momentum, \ie like the square root of its energy. Therefore,
we can view our rotor model as some kind of ``infinite pinning''
limit. 
\end{remark}

We make the following non-degeneracy assumption (clearly satisfied for
\eg $w_1=w_3=\sin$):
\begin{assumption}\label{as:assumptioncoupling} There is at least one $b\in
\{1,3\}$
such that for each $s\in \mathbb T$, at least one of the derivatives
$w_b^{(k)}(s), k\geq 1$
is non-zero.
\end{assumption}

For all initial conditions $x\in \Omega$ and all times $t\geq 0$, we denote by
$P^t(x, \argcdot)$ the transition probability of the Markov process associated
to \eref{eq:sde}. Since the coefficients of the SDE \eref{eq:sde} are globally
Lipschitz, the solutions are almost surely defined for all times and all initial
conditions, so that $P^t(x, \argcdot)$ is well-defined as a probability measure
on $(\Omega, \mathcal B)$.

We now introduce the main theorem, in which we write
$$
\|\nu\|_{f} = \sup_{|h|\leq f}\int_\Omega h \dd \nu
$$
for any continuous function
$f > 0$ on $\Omega$ and any signed measure $\nu$ on $(\Omega, \mathcal B)$.

\begin{theorem}\label{thm:mainthm} Under \aref{as:assumptioncoupling}, the
following holds for the Markov process defined by \eref{eq:sde}:
\begin{myenum}
	\item[(i)] The transition kernel $P^t$ has a  density
	$p_t(x, y)$ in $\mathcal C^\infty((0, \infty)\times \Omega\times \Omega)$.
	\item[(ii)] The process admits a unique invariant measure
          $\pi$, which has a 
	smooth density.
	\item[(iii)] For all sufficiently small $\beta > 0$ and all $ \beta' \in [0, \beta)$,
	there are constants $C, \lambda> 0$ such that for all $t\geq 0$ and all 
    $x = (q_1, q_2, \dots, p_3)\in\Omega$,
	\begin{equ}
		\|P^t(x, \argcdot) - \pi \|_{e^{\beta'H}} \leq C(1+p_2^2)e^{\beta H(x)}e^{-\lambda t^{1/2}}~.
	\end{equ}
\end{myenum}
\end{theorem}

\begin{remark}
If both heat baths are at the same temperature, say, $T_1 = T_3 = T > 0$, and the
forces $\tau_1$ and $\tau_3$ are zero, then the
system is at thermal equilibrium and the Gibbs measure with density proportional
to $e^{-H/T}$ is invariant. Indeed, one easily checks that this density verifies
the stationary Fokker-Planck equation $L^*e^{-H/T} = 0$, where $L^*$ is the
 formal adjoint of the generator $L$ introduced below.
\end{remark}

\begin{remark} In fact, the results we prove here apply with hardly any
modification to the ``star'' configuration with one central rotor interacting
with $m$ external rotors, which in turn are coupled to heat baths (\ie $m+1$
rotors and $m$ heat baths). 

In addition, some studies
(\eg \cite{iacobucci_negative_2011}) consider chains with fixed boundary conditions.
For the left end of the chain, this corresponds to adding a ``dummy'' rotor
 0 which does not move but interacts with rotor 1. This is covered by our
theory by adding some contribution to the pinning potential $U_1$. The same
applies to the right end and $U_3$.
\end{remark}

Chains of rotors provide toy models for the study of non-equilibrium
statistical mechanics. In \cite{iacobucci_negative_2011} long chains have been
studied numerically, and it appears that even when the external temperatures are
different
and external forces are applied, local thermal equilibrium is satisfied in the
stationary state in the limit of infinitely long chains. This stationary state
may have some surprising features, like a large amount of energy
in the bulk of the chain when the
boundary conditions are properly chosen.
In our case of course we are far from local thermal equilibrium, 
since we only study
systems made of three rotors. We will present some numerical simulations
of our system in \sref{sec:numerics}, highlighting some interesting properties 
of the stationary state.

What corresponds here to the breathers observed in other models is the
situation where the energy of the system is very large and mostly concentrated
in the middle rotor. The middle rotor then spins very rapidly, 
and the interaction forces oscillate so fast that they have very little
net effect. In this case, the middle rotor effectively decouples
from the rest of the system, and the main difficulty is to show that
its energy eventually decreases with some well-controlled bounds.

The idea used in \cite{hairer_slow_2009} for the chain
of three pinned oscillators is to average the oscillatory forces, and 
exhibit a negative feedback in the regime where the breather dominates the
dynamics.
The proof of \tref{thm:mainthm} in the present paper is based on a
systematisation
of this idea, as explained in \sref{sec:average}.

The paper is structured as follows: In \sref{sec:ELF} we introduce
a sufficient condition for subgeometric ergodicity from 
\cite{douc_subgeometric_2009}.
In \sref{sec:p2dynamics} we study the behaviour
of the middle rotor. In \sref{sec:effectivedyntoLyapunov} we show how to use
the study of $p_2$ to get a Lyapunov function.
In \sref{sec:controlirred} we provide the necessary technical input to the
theorem of \cite{douc_subgeometric_2009}. Finally, we illustrate numerically
some properties of the non-equilibrium steady state in \sref{sec:numerics}.

\section{Ergodicity and Lyapunov functions}\label{sec:ELF}

The proof of \tref{thm:mainthm} relies on the results of
\cite{douc_subgeometric_2009} which in turn are based on
the theory exposed in \cite{meyn_markov_2009}. The theory of \cite{meyn_markov_2009}
shows that one can prove the
ergodicity of an irreducible Markov 
process and estimate the rate of convergence toward its
invariant measure if one has a good control of the return
times of the process to particular sets, called {\em  petite sets}.
A set $K$ is petite if there exist a probability measure $a$ on $[0,\infty)$
and a non-zero measure $\nu_a$ on $\Omega$ such that for all $x\in K$ one has
$\int_0^\infty P^t(x,\argcdot)a(\dd t) \geq \nu_a(\argcdot)$. 
In the case we are interested in, control arguments and the hypoellipticity of
the generator imply that each compact set is petite
(see \sref{sec:smooth} for a proof of this property).

Let $L$ be the infinitesimal generator of the process, \ie the second-order
differential operator
\begin{equ}
L = \sum_{i=1}^3 \left(p_i \partial_{q_i} - u_i(q_i)\partial_{p_i} \right) 
+ \sum_{b}[w_b(q_2-q_b)(\partial_{p_b}-\partial_{p_2}) + \tau_b
\partial_{p_b} - \gamma_b p_b \partial_{p_b} + \gamma_b T_b \partial_{p_b}^2]~.
\end{equ}
Recall that for any sufficiently regular function $f$ we have 
$Lf(x)=\frac{\dd}{\dd t}\left.\left[ \int f(y) P_t(x,\dd y)\right]\right|_{t=0}$.

A classical way to control the return times to a petite set is to make use of
Lyapunov functions.
We call {\em Lyapunov function} a smooth function
 $V: \Omega \mapsto [1, \infty)$ with compact level
sets (\ie due to the structure of $\Omega$,
a function such that $V( q,  p)\to\infty$ when $\| p\| \to \infty$)
such that for all $x\in \Omega$,
\begin{equ}\label{eq:driftcondition}
(L V)(x) \leq C\ind_{K}(x) -\phi \circ V(x)~,
\end{equ}
where $C$ is a constant, $\phi : [1, \infty) \to \mathbb (0, \infty)$ is an
increasing function, and $K$ is a petite set. If one can find such a function,
and prove that some skeleton $P^\Delta (\Delta>0)$ is $\mu$-irreducible for some
measure
$\mu$ (\ie $\mu(A)>0$ implies that for all $x\in \gO$ there exists $k\in
\bbN$ such
that $P^{k \Delta}(x,A)>0$),
then the Markov process is 
ergodic, with rate depending on $\phi$. In the case where
$\phi(V) \propto V^\rho$,
the convergence is geometric if $\rho=1$
and polynomial if $\rho < 1$ (see \cite{meyn_stability_1993,douc_subgeometric_2009}).
In this paper, we obtain $\phi(V) \sim  V/\log V$. 

We rely on the work of
Douc, Fort and Guillin \cite{douc_subgeometric_2009}, which gives a sufficient
condition for subgeometric ergodicity of continuous-time Markov processes.
We give here a simplified version of their result, adapted to our
purpose. This statement is based on Theorem 3.2 and
Theorem 3.4 of \cite{douc_subgeometric_2009}.

\begin{theorem}[Douc-Fort-Guillin (2009)]\label{th:douc}
Assume that the process has an irreducible skeleton and that there exist a smooth
function $V:\Omega\rightarrow [1,\infty)$ with
$V( q,  p)\to\infty$ when $\| p\| \to \infty$, an increasing, differentiable,
concave function $\phi:[1,\infty)\rightarrow(0,\infty)$, a petite set $K$, and a
constant $C$ such that \eref{eq:driftcondition} holds. Then the process admits
a unique invariant measure $\pi$, and for each $z\in [0,1]$, there exists a
constant $C'$ such that for all $t\geq 0$ and all $x\in \Omega$,
\begin{equ}
\|P^t(x, \argcdot) - \pi \|_{(\phi \circ V)^{z}} \leq	g(t) C' V(x)~,
\end{equ}
where $g(t)=(\phi\circ H_\phi^{-1}(t))^{z-1}$,
with $H_\phi(u)=\int_1^u \frac{\dd s}{\phi(s)}$.
\end{theorem}

When $z=0$,
we retrieve the total variation norm
$\|P^t(x, \argcdot) - \pi \|_{\mathrm{TV}}$ and the rate is the fastest. Increasing
$z$ strengthens the norm but slows the convergence rate down. When $z=1$, the
norm is the strongest, but no convergence is guaranteed since $g(t)\equiv 1$.

The core of the paper is devoted to the construction of a Lyapunov function
such that \eref{eq:driftcondition} is satisfied with $\phi(s) \sim s/\log s$,
and a set $K$ which is compact and therefore petite.
This yields a stretched exponential
convergence rate (see \eref{eq:convrate}). The existence 
of an irreducible skeleton required by \tref{th:douc}
and the fact that every compact set is
petite are proved in \sref{sec:controlirred}.

One might at first think that a Lyapunov function is simply given by the
Hamiltonian $H$. Unfortunately, this is not the case, as
\begin{equ}\label{eq:LH}
LH = \sum_{b} \big(\tau_b p_b+\gamma_b(T_b-p_b^2)\big)~,
\end{equ}
where the right-hand side remains positive when $p_1, p_3$ are small
and $p_2\to \infty$. Thus, there
is no bound of the form \eref{eq:driftcondition} for $H$. The
same problem occurs if we take any function $f(H)$ of the energy.

In order to find a {\it{bona fide}} Lyapunov function, we will need
more insight into how 
fast all \emph{three} momenta decrease. The equality \eref{eq:LH} suggests
that $p_1$ and $p_3$ will not cause any problem.
In fact, we have for $b=1,3$, that 
$$
Lp_b = -\gamma_b p_b + w_{b}(q_2-q_b) - u_b(q_b) + \tau_b~.
$$
Since $w_{b}(q_2-q_b) - u_b(q_b) + \tau_b$ is bounded,
$|p_b|$ essentially decays at exponential rate when it is large.
This is of course due to the friction terms that act on
$p_1$ and $p_3$ directly.
Such a result does not hold for $p_2$. In fact, the decay of $p_2$ is
much slower. Our main insight is that in a sense
\begin{equ}\Label{e:nminus4}
Lp_2 \sim -cp_2^{-3}~.
\end{equ}
The proof of such a relation occupies a major part of this paper.
As indicated earlier, this very slow damping of $p_2$ comes
from the lack of effective interaction when the forces oscillate 
very rapidly.
Once we have gained enough
understanding of the dynamics of $p_2$, we will be able to construct a Lyapunov
function, whose properties are summarised in
\begin{proposition}\label{prop:Lyapunov}
For all sufficiently small $\beta > 0$, there is a function
$V: \Omega \to [1,\infty)$ satisfying the two
following properties:
\begin{myenum}
\item There are positive constants $c_1, c_2$ such that
\begin{equ}\Label{eq:ineqvn}
	1 + c_1 e^{\beta H}\leq V\leq c_2(1+p_2^2) e^{\beta H}~.
\end{equ}
\item There are positive constants $c_3, c_4$ and a compact set $K$ such that 
\begin{equ}\Label{eq:ineqLvn}
	LV \leq c_3 \ind_K - \phi(V) ~,
\end{equ}
where $\phi: [1, \infty) \to (0, \infty)$ is defined by\footnote{The 2 in the denominator ensures that $\phi$
is concave and increasing on $[1, \infty)$, as required in \tref{th:douc}.}
\begin{equ}\label{eq:defphi}
\phi(s) = \frac {c_4\,s}{2+\log(s)}~.
\end{equ}
\end{myenum}
\end{proposition}

The way we construct the Lyapunov function is somewhat different
from that of \cite{hairer_slow_2009}. There, it is obtained starting
from some power of the Hamiltonian and then adding corrections by 
an averaging technique similar to ours (see \rref{rem:averagingsimilarity}).
Here, we first average the dynamics of $p_2$ and then use the result to
construct a Lyapunov function that essentially grows exponentially with the
energy. This gives a stretched exponential rate of convergence
instead of a polynomial rate as in
\cite{hairer_slow_2009}. The present method can in principle be applied to the
model of \cite{hairer_slow_2009} (see also \cite{hairer_how_2009}).

We now show how the main results follows.

\begin{proof}[Proof of \tref{thm:mainthm}] The conclusions of \tref{thm:mainthm}
immediately follow from \tref{th:douc}, \pref{prop:Lyapunov}, the technical results stated in
\pref{prop:controlstart}, and the following two observations. Consider
$0\leq \beta' < \beta$ and choose $z\in (0,1)$ such that  $\beta ' < z\beta$.
First, the function $\phi$ defined in \eref{eq:defphi} yields, in the notation of \tref{th:douc},
a convergence rate
\begin{equ}[eq:convrate]
g(t) = (\phi\circ H_\phi^{-1}(t))^{z-1} \leq c e^{-\lambda t^{1/2}}
\end{equ}
for some $c, \lambda > 0$.
Indeed, we have $H_\phi(u)= \frac 1{c_4}\int_1^u \frac {2+\log s}s \dd s =
 \frac {1}{2c_4} (\log u)^2+\frac 2{c_4} \log u$, so that
$H_\phi^{-1}(t) = \exp({(2c_4t +4)^{1/2}-2})$ and 
$(\phi\circ H_\phi^{-1}(t)) =(2c_4t +4)^{-1/2} \exp((2c_4t +4)^{1/2}-2) \geq Ce^{C't^{1/2}}$
for some $C, C'>0$. Thus, \eref{eq:convrate} holds with $\lambda = (1-z)C'$.
Secondly, by \pref{prop:Lyapunov} (i), and since  $\beta ' < z\beta$, we observe
that $e^{\beta'H} \leq c(\phi \circ V)^z$
for some constant $c>0$, so that
$\|\argcdot\|_{e^{\beta'H}} \leq c\, \|\argcdot\|_{(\phi \circ V)^{z}}$.
\end{proof}

\section{Effective dynamics for the middle rotor}\label{sec:p2dynamics}
The hardest and most interesting part of the problem
is to determine how $p_2$ decreases when it is
very large.\footnote{To simplify notation, we say $p_2$ is large, but
we always really mean that $|p_2|$ is large.} In this section, we obtain
some asymptotic, effective
dynamics for $p_2$ when $p_2\to \infty$. 

\subsection{Expected rate}

Before we start making any proof, we can get a hint of how $p_2$ decreases in
the regime where $p_2$ is very large and both $p_1, p_3$ are small. Assume for
simplicity that $u_i \equiv 0$ and that $W_b(s) = -\kappa  \cos(s)$ so that
$w_b(s) = \kappa \sin(s)$. In the regime of interest, we expect the middle
rotor to decouple, so that $p_2$ will evolve very slowly.
We will consider the system over times that are small enough for $p_2$ to remain
almost constant (say equal to $\omega$), 
but large enough for some
``quasi-stationary'' regime to be reached. 
The reader can think of $\omega $ as being the ``initial'' value of $p_2$.
For $b=1,3$, we expect $p_b$ to be
well approximated, at least qualitatively, by the equation
\begin{equ}\Label{eq:smallmodel}
\dd p_b = \kappa \sin(\omega  t)\dt - \gamma_b p_b\dt  + \sqrt{2\gamma_b
T_b}\,\dd
B^b_t~,
\end{equ}
whose solution is
\begin{equ}
\begin{aligned}
	p_b(t) &= \kappa \frac {\gamma_b \sin(\omega t)-\omega\cos(\omega
		t)}{\gamma_b^2+\omega^2}  + \sqrt{2\gamma_b T_b}\int_0^t e^{\gamma_b(s-t)}\dd
	B^b_s\\
	&= -\kappa \frac {\cos(\omega t)}{\omega}  + \sqrt{2\gamma_b T_b}\int_0^t
	e^{\gamma_b(s-t)}\dd B^b_s+  \mathcal O\!\left(\frac 1 {\omega^2}\right)~.
\end{aligned}
\end{equ}
We have neglected the exponentially decaying part $ p_b(0)e^{-\gamma_b
t}$ since we assume that a quasi-stationary regime is reached. By
\eref{eq:LH}, the rate of energy flowing {\em into} of the system at $b$ is
$\gamma_b (T_b - p_b^2)$. Squaring $p_b$ and taking expectations, what remains
is
\begin{equ}
\begin{aligned}
	\mathbb E p_b^2(t) &= \kappa^2 \frac{ \cos^2(\omega t)}{\omega^2 } +
	2\gamma_bT_b \mathbb E\left(\int_0^t e^{\gamma_b(s-t)} \dd B^b_s \right)^2 + 
	\mathcal O\!\left(\frac 1 {\omega^3}\right)\\
	&= \kappa^2 \frac{ \cos^2(\omega t)}{\omega^2 } +(1-e^{-2\gamma_b t})T_b + 
	\mathcal O\!\left(\frac 1 {\omega^3}\right)~,
\end{aligned}
\end{equ}
where we have used the It\=o isometry $\mathbb E(\int_0^t e^{\gamma_b(s-t)} \dd
B^b_s)^2 = \int_0^t e^{2\gamma_b(s-t)} \dd s $. Neglecting again an
exponentially decaying term, we obtain
\begin{equ}\label{eq:decayapprox}
\frac \dd {\dt} \mathbb E H(t) = \sum_{b}\mathbb E\left(\gamma_b (T_b -
p_b^2(t))\right) \sim
-\sum_{b}\gamma_b \kappa^2\frac{\cos^2(\omega t)}{\omega^2}~.
\end{equ}
Since $\cos^2(\omega t)$ oscillates very rapidly around its average $1/2$,
we expect to see an effective contribution $ -\frac {\gamma
\kappa^2}{2\omega^2}$. This approximation was obtained by assuming that $p_2$ is
almost constant and equal to $\omega $. Now, when $p_2$ is very large, the
energy $H$ is dominated by
the contribution $\frac 12 p_2^2$, so that we expect to have $\frac \dd {\dt}
\mathbb EH \sim  p_2
\frac \dd {\dt} \mathbb Ep_2$. Comparison with \eref{eq:decayapprox} leads to
\begin{equ}\Label{eq:expectedorder}
\frac \dd {\dt} \mathbb Ep_2 \sim -\frac 1 {p_2^3}\sum_{b}\frac {\gamma_b
\kappa^2}{2}~.
\end{equ}
We will obtain this result rigorously in \pref{p:mainprop}. 

\subsection{Notations}

Let $\Omega^\dagger=\{( q, p)\in \Omega:\, p_2\neq 0\}$. 
We denote throughout by $X_t = ( q(t), 
p(t))$ the solution of the stochastic differential equation
\eref{eq:sde} with initial condition $X_0 = ( q(0),  p(0))$.
For now, we restrict ourselves to  
$X_0\in \Omega^\dagger$ since we aim to obtain
an effective dynamics for the middle rotor by performing an
expansion in negative powers of $p_2$.
Remark that since $\frac{\dd}{\dd t} p_2$ is bounded, there is for each initial
condition 
$X_0 \in \Omega^\dagger$ a deterministic time $t^* > 0$ (proportional to
$|p_2(0)|$)
such that $X_t \in \Omega^\dagger$ for all $t\in [0, t^*)$ and 
all realisations of the random noises.
To define a smooth Lyapunov function on the whole space
 $\Omega$ we will perform a regularisation in \sref{sec:Lyapunovconstr}.

\begin{definition}\label{def:ensembleU}
We let $\cU$ be the set of stochastic processes $u_t$
which are solutions of an SDE of the form
\begin{equ}\label{e:long}
	\dd u_t = f_1(X_t)\dd t + f_2(X_t)\dd B^1_t+f_3(X_t) \dd B^3_t~,
\end{equ}
for some functions
$f_i:\Omega\to\bbR$.
\end{definition}

{\bf{Notation}}: In the sequel, we write 
\begin{equ}
	\dd u_t = f_1\dd t + f_2\dd B^1_t+f_3 \dd
        B^3_t
\end{equ}
instead of \eref{e:long}.

For any smooth function $h$ on $\Omega$, the stochastic process $h(X_t)$ is in $\cU$ by
the It\=o formula (see below). Without further mention, we will both see $h$ as
a
function on $\Omega$ and as the stochastic process $h(X_t)$. When referring to
the stochastic process, we shall write simply $\dd h$ instead of $\dd h(X_t)$.
Of course, only very few processes in $\cU$ can be written in the form $h(X_t)$
for some function $h$ on $\Omega$.

The variables $p_2$ and $q_2$ will play a special role, as we are merely
interested in the regime where $p_2$ is very large. For any function $f$ over
$\Omega$ we call the quantity
$$
\avg f = {\frac 1{2\pi}}\int_{0}^{2\pi }f\, \dd q_2
$$
the {\em $q_2$-average} of $f$ (or simply the {\em average} of $f$), which is a
function of $ p$, $q_1$ and $q_3$ only.
\begin{assumption}\label{ass:moyennenulle} We assume 
$$\avg{U_2} = 0\qquad \text{and} \quad\avg{W_b} = 0,\qquad b=1,3 ~.$$
\end{assumption}
This assumption merely fixes the
additive constants of the potentials and therefore
results in no loss of generality.

For conciseness, we shall omit the arguments of the potentials and forces,
always assuming that
\begin{align*}
W_b &= W_b(q_2-q_b)~, & w_b &= w_b(q_2-q_b)~, & b&=1,3~,\\
U_i &= U_i(q_i)~, & u_i &= u_i(q_i)~, & i&=1,2,3~.
\end{align*}

To simplify the notations, we also introduce the potentials $\Phi_1$, $\Phi_2$,
$\Phi_3$ associated to the three rotors, and the corresponding forces $\phi_1,
\phi_2, \phi_3$ defined by
\begin{align}\label{eq:defphis}
\Phi_b &= W_b + U_b~ ,  & \phi_b &= -\partial_{q_b}\Phi_b = w_b - u_b~, &
b=1,3~,\notag\\ 
\Phi_2 &= W_1 + W_3 + U_2~, & \phi_2 &= -\partial_{q_2}\Phi_2 = -w_1 - w_3 -
u_2~.
\end{align}
Of course, $\Phi_i$ and $\phi_i$ are functions of $  q$ only. With these
notations the dynamics reads more concisely
\begin{align*}\Label{eq:sdere}
\dd q_i&=p_i\dt~, & i=1,2,3~,\\
\dd p_2&=\phi_2 \dt~,\\
\dd p_b&=\Big(\phi_b +\tau_b-\gamma_b p_b\Big)\dt +\sqrt{2\gamma_b T_b}\,\dd
B^b_t, & b=1,3~.
\end{align*}
We will mainly deal with functions of the form $p_2^\ell p_1^np_3^m g(
q)$ and their linear combinations. We therefore introduce the notion
of degree.
\begin{definition}\label{d:degree}We say that a function $f$ on 
$\Omega^\dagger$ has {\em degree} $\ell
\in \mathbb Z$ if it can be written as a  {\em finite} sum of elements of the
kind $p_2^\ell p_1^np_3^m g( q)$ for some $n, m \in \mathbb N$
and a smooth function $g:\bbT^3\to\bbR$. Moreover,
we denote
$$
\bigoh{\ell}
$$
a generic expression of order at most $\ell $ (which can vary from line to
line),
\ie a finite sum of functions of degree $\ell$, $\ell-1$,
$\ell-2$,\dots.
\end{definition}
We have by the It\=o formula that for any smooth function $f$ on $\Omega$
\begin{equ}\Label{eq:itoformshort}
\begin{aligned}
	\dd f & = \sum_{i=1}^3 \left(   \mpart{f}{q_i} \dd q_i +\mpart{f}{p_i} \dd
	p_i\right) + \sum_{b}\gamma_b T_b \frac{\partial^2 f}{\partial p_b^2}\dt\\
	& = \dd^+ f + \dd^0f + \dd^- f~,\\
\end{aligned}
\end{equ}
where
\begin{equ}\label{e:ddd}
\begin{aligned}
	\dd^+ f &=  p_2 \mpart{f}{q_2}\dt~,\\
	\dd^- f &=  \phi_2\mpart{f}{p_2}\dt~,\\
	\dd^0 f &=  \sum_{b}\left(p_b \mpart{f}{q_b} +(\phi_b+ \tau_b- \gamma_b
	p_b)\mpart{f}{p_b}  + \gamma_b T_b \frac{\partial^2f}{\partial
		p_b^2}\right)\dt   +  \sum_{b}
	\sqrt{2\gamma_b T_b}\mpart{f}{p_b} \dd B^b_t~.
\end{aligned}
\end{equ}
(By the discussion following \dref{def:ensembleU}, $f$, its partial derivatives,
$p_2$ and the functions $\phi_i$ in this SDE are evaluated on the trajectory
$X_t$.)
Observe that when acting on a function of degree $\ell$, the contribution $\dd
^+$ increases the degree of $p_2$ by one, while  $\dd^0$ and $\dd^-$
respectively leave it unchanged and decrease it by one. In this sense, we will
see $\dd^+$ as the ``dominant'' part of $\dd$.

\subsection{General idea}
\label{subs:generalidea}
In this section we introduce the main idea, which consists in successively
removing oscillatory terms order by order in the dynamics of $p_2$. We
perform here 
the first step of the method in a somewhat naive, but pedestrian way. In the
next two sections, we systematise the method and apply it.

We begin by looking at the equation
\begin{equ}\label{eq:startingpointp2}
\dd p_2 = \phi_2 \dt~.
\end{equ}
When $p_2$ is large while $p_1$ and $p_3$ are small, the right-hand
side is highly 
oscillatory and its time-average is almost zero, since $\avg{\phi_2} = 0$. We
will proceed to a change of variable in order to ``see through'' this
oscillatory term. 
 
We first make the relation between the time-average and the $q_2$-average more
precise. Consider some function $g$ on $\Omega$. In the regime where $p_2$ is
very large and $p_1$, $p_3$ are small, the only fast variable is 
$q_2$.
Now consider some interval of time $[0,T]$ short enough so that the other
variables do not change significantly, but still large enough for $q_2$ to swipe
through $[0, 2\pi)$ many times.
We have in that case $q_2(t) \sim q_2(0) + p_2(0) t$ (remember that $q_2$ is
defined modulo $2\pi$) and
\begin{equ}\label{eq:approxp2t}
g( q(t),  p(t)) \sim g\big(q_1(0), q_2(0)+p_2(0) t, q_3(0),
 p(0)\big)~,
\end{equ}
so that the time-average of $g$ is expected to be very close to the
$q_2$-average $\avg g$.

Now, we want to estimate $p_2(t) =
\int_{0}^t \phi_2( q(s)) \dd s$ in this situation. Approximating $\phi_2$ as in
\eref{eq:approxp2t} and integrating formally with respect to time (remember that
$\phi_2 = -\partial_{q_2} \Phi_2$) leads naturally to the decomposition
\begin{equ}\label{eq:modifierp2}
p_2 = \bar p_2 -\frac {\Phi_2(q)}{p_2}~,
\end{equ}
which consists in writing $p_2$ as sum of an oscillatory term $
{\Phi_2}/{p_2}$ which is supposed to capture ``most'' of the oscillatory
dynamics, and some (hopefully) nicely behaved ``slow'' process $\bar p_2$. And
indeed, if we differentiate \eref{eq:modifierp2} we get
\begin{equ}\label{eq:decompp2naive}
\begin{aligned}
	\dd \bar p_2  &=\dd\left( p_2 + \frac {\Phi_2}{p_2}  \right)\\&= \dd p_2 +
\dd^+
	\frac{\Phi_2}{p_2} + \dd^0\frac{\Phi_2}{p_2} + \dd^-\frac{\Phi_2}{p_2}\\
	& =\phi_2\dt - \phi_2 \dt -\left(\frac{p_1w_1}{p_2}+ \frac{p_3w_3}{p_2}
	\right)\dt  -  \frac {\phi_2\Phi_2}{p_2^2}\dt\\
	& = \hphantom{\phi_2\dt - \phi_2 \dt} -\left(\frac{p_1w_1}{p_2}+
\frac{p_3w_3}{p_2} \right)\dt  -  \frac
	{\phi_2\Phi_2}{p_2^2}\dt~.
\end{aligned}
\end{equ}
As a result, we have a new process $\bar p_2$ which is asymptotically equal to
$p_2$ in the regime of interest, and whose dynamics involves only terms that are
small when $p_2$ is large, so that $\bar p_2$ is indeed a slow variable.
Observe that the choice of adding ${\Phi_2}/{p_2}$ to $p_2$ has the
effect that $\dd^+( {\Phi_2}/{p_2}) = -\phi_2\dt$, which precisely cancels
the right-hand side of \eref{eq:startingpointp2} while the remaining
terms have negative powers of $p_2$. This observation is the starting
point of the systematisation of the method.

Unfortunately, \eref{eq:decompp2naive} is not good enough to
understand how $\bar 
p_2$ (and therefore $p_2$) decreases in the long run, since the dynamics
\eref{eq:decompp2naive} of $\bar p_2$ still involves oscillatory terms. The idea
is therefore to eliminate these oscillatory terms by
absorbing them into a further change of variable $\bar p_2 = \bar{\bar p}_2 + G$
for some suitably chosen $G$. The result is that $\dd \bar{\bar p}_2$ is
a sum of terms of degree $-2$ at most, which turn out to be still oscillatory.
This procedure must then be iterated, successively eliminating oscillatory terms order by
order, until we get some dynamics that has a non-zero average (which happens after
finitely many steps).
We will follow this idea, but in a way
that does not require to write the successive changes of variable explicitly.
More precisely, we will prove 

\begin{proposition}\label{p:mainprop}
There is a function $F = \frac {\Phi_2(q)}{p_2} + \bigoh{-2}$ such that whenever
$p_2(t)  \neq 0$ the
process $\tilde p_2(t) = p_2(t) + F(X_t)$ satisfies
\begin{equ} \label{e:p1nm4equiv}
	\dd \tilde p_2(t) =  a(X_t) \dt + \sum_b  \sigma_b(X_t)\dd B^b_t ~,
\end{equ}
with
\begin{equ}
	\begin{aligned}
		a( q,  p) &= -     \frac { {\gamma_1  \avg{W_1^2} + \gamma_3 \avg{W_3^2}}}
{p_2^3} 
		+ \bigohneg{4} ~ ,\\
		\sigma_b( q,  p) & = \frac{\sqrt{2\gamma_bT_b}{W_b}}{
			p_2^2}+\bigohneg{3}~, \qquad b=1,3~.
	\end{aligned}
\end{equ}
(By \aref{ass:moyennenulle}, no arbitrary additive constant appears in
$\avg{W_1^2} $ and $\avg{W_3^2}$.)
\end{proposition}
The next two sections are devoted to proving \pref{p:mainprop}.

\subsection{Averaging}\label{sec:average}

The crux of our analysis is to average oscillatory terms in the dynamics. This is a 
well known problem in differential equations. In classical averaging theory
\cite{sanders2007averaging,vela_averaging_2003}, it is an {\em external} small
parameter $\varepsilon$ that gives the time scale of the fast variables. Here,
the role of $\varepsilon$ is played by $ 1/{p_2}$, which is a dynamical
variable. We develop an averaging theory adapted to this case, and
also to the stochastic nature of the problem.

The starting point is as follows. Imagine that for a function $h$ on $\Omega$ we
find an expression of the kind
\begin{equ}\label{eq:dh}
\dd h = f\dt + \dd r_t~,
\end{equ}
for some function $f = f(X_t)$ of degree $\ell$ and some stochastic process $r_t
\in \cU$ (see \dref{def:ensembleU}) which denotes the part of the dynamics that
we do not want to interfere with. Thinking of $f(X_t)$ as a highly oscillatory
quantity when $p_2$ is very large, we would like to write $h = \bar h + F$ for
some small function $F$ on $\Omega$ such that 
\begin{equ}\label{eq:simplifierh}
\dd \bar h   =  \avg f \dt +  \dd r_t  + \text{small corrections}\,,
\end{equ}
where the notion of {\em small} will be made precise in terms of powers of
$p_2$. That is, we want to find some $\bar h$ close to $h$, such that its
dynamics involves, instead of $f\dt$, the $q_2$-average $\avg f\!\dt$ plus some
smaller corrections. In other words, we are looking for some $F$ such that
$$
\dd F = \dd(h-\bar h) = (f-\avg f)\dt + \text{small corrections}\,.
$$
Remembering that in terms of powers of $p_2$,  $\dd^+$ is the dominant part of
$\dd$, the key is to find some $F$ such that $\dd^+ F = (f-\avg
f)\dt$. If we write $L^+=p_2\partial_{q_2}$ , we have
$\dd^+ F=L^+ F \dd t$. Thus, we really need to invert $L^+$ (which is in fact
the dominant part of the generator $L$ when $p_2$ is large). 

We call here $\KK$ the space of smooth 
functions $\Omega^\dagger \to \mathbb R$, and we denote by $\KK_0$ the space
of functions $f\in \KK$ such that $\avg{f}=0$.
Note that $L^+$ maps $\KK$ to $\KK_0$ since for all $f\in\KK$, we have by
periodicity
$$
\avg{L^+f}=p_2\avg{\partial_{q_2} f}=0~.
$$
We can define a right inverse $(L^+)^{-1} : \KK_0 \to \KK_0$ by letting for all
$g\in \KK_0$
\begin{equ}\Label{eq:Flmoin1}
 (L^+)^{-1}g = \frac 1 {p_2}\left( \int g \,\dd q_2+c( p, q_1, q_3)\right)~,
\end{equ}
where the integration ``constant'' $c( p, q_1, q_3)$ is
uniquely defined by requiring that $\avg{(L^+)^{-1}g} = 0$. 

This leads naturally to the following
\begin{definition} For any function $f\in\KK$,
  we define the operator  $Q: \KK \to \KK_0$ by
\begin{equ}\Label{eq:defLpm1}
	Qf = (L^+)^{-1} (f-\avg{f})~.
\end{equ}
\end{definition}
\begin{remark}\hfill
\begin{myitem}
\item If $f$ is a function of degree $\ell$, then $Qf$ is of degree $\ell-1$.
\item By construction,
\begin{equ}\label{eq:propQf}
	\dd(Qf) = (f-\avg f)\dt  + \dd^0(Qf) + \dd^-(Qf)~.
\end{equ}
\item Moreover, by definition, $Qf$ is the only function such that
\begin{equ}\label{eq:redefQf}
	\partial_{q_2}(Qf) =\frac{f-\avg{f}}{ p_2} \quad \text{and}\quad 
	\avg{Q f}=0~.
\end{equ}
\end{myitem}
\end{remark}

Therefore, if \eref{eq:dh} holds for some $f$ of degree $\ell$, then we obtain a
quantitative expression for \eref{eq:simplifierh}, namely 
\begin{equ}\Label{eq:simplifierh2}
	\dd (h - Qf) =\avg f \dt +   \dd r_t  - \dd^0(Qf) - \dd^-(Qf)~,
\end{equ}
where the corrections are small in the sense that $Qf$, $\dd^0(Qf)$
and $\dd^-(Qf)$ have degree respectively $\ell-1$, $\ell-1$ and $\ell-2$. 

\begin{remark}
Observe that \eref{eq:modifierp2} can be written
now as $p_2 = \bar p_2 + Q\phi_2$, since $Q\phi_2 = -{\Phi_2}/{p_2}$. Thus,
the ``naive'' correction we added in \eref{eq:modifierp2} also follows
from the systematic method we have just introduced.
This is no surprise: the naive correction in \eref{eq:modifierp2} was motivated
by the approximation \eref{eq:approxp2t} in which only $q_2$ moves,
which corresponds to considering only $\dd^+$.
\end{remark}

\begin{remark}\label{rem:averagingsimilarity}Our averaging procedure is inspired by techniques of
\cite{hairer_slow_2009}. There, the equivalent of
$L^+$ is the generator 
$-q_2^{2k-1}\partial_{p_2} + p_2 \partial_{q_2}$ of the free dynamics of the
middle oscillator, where  ${q_2^{2k}}/({2k})$ is the pinning potential. In
their case, one cannot explicitly invert $L^+$, but one can show that
$(L^+)^{-1}$ basically acts as a division by $E_2^{\frac 1 2-\frac 1{2k} }$,
where $E_2$ is the energy of the middle oscillator. Again, taking formally the
limit $k\to \infty$, one obtains that $(L^+)^{-1}$ acts as a division by $
\sqrt{E_2}$, much like in our case where $(L^+)^{-1}$ acts as
a division by ${p_2} \sim \sqrt{E_2}$.
\end{remark}

We now restate our averaging method as the following lemma, which follows from
a trivial rearrangement of the terms in \eref{eq:propQf}.  

\begin{lemma}\label{lem:remplaceravg} (Averaging lemma) Consider some function
$f = \bigoh{\ell}$ for some $\ell\in\bbZ$. Then
\begin{equ}
        f\dt  =  \langle f\rangle \dt - \dd^0\left( Qf\right)-
\dd^-\left(Qf\right)+ \dd(Qf)~,
\end{equ}
where $\dd^0(Qf)$ is of degree $\ell-1$ at most and $\dd^-(Qf)$ is
of degree $\ell-2$ at most. 
\end{lemma}

We now prove \pref{p:mainprop} by using \lref{lem:remplaceravg} repeatedly.

\subsection{Proof of \pref{p:mainprop}}\label{sec:proof prop}

We make the following observations,
which we will use without reference. For any function $f$ on $\Omega^\dagger$ 
that is smooth in $q_2$, we have by periodicity
\begin{equ}\label{eq:dermoyennenulle}
\avg{ \partial_{q_2}f} = 0~.
\end{equ}
Moreover, if $g$ is another such function, then we can integrate by parts to
obtain 
\begin{equ}\Label{eq:parparties}
\avg {(\partial_{q_2}f) g}= -\avg{ f \partial_{q_2}g}~.
\end{equ}
Furthermore, we have by \aref{ass:moyennenulle}, \eref{eq:defphis} and
\eref{eq:dermoyennenulle} that
\begin{equ}\Label{eq:toutmoyennenulle}
\avg{W_b } =\avg{w_b } =\avg{\Phi_2 } =\avg{\phi_2 } = 0~.
\end{equ}

We start by doing again the
first step, which we did in \sref{subs:generalidea}, but this time using the new
toolset. In order to average the right-hand side of
\begin{equ}
	\dd p_2 = \phi_2 \dt~,
\end{equ}
we use \lref{lem:remplaceravg} with $f = \phi_2$, which is of order $0$. We have
$\avg f = 0$ and $Qf =  -{\Phi_2}/{p_2}$ (by definition of $\phi_2$ and
$\Phi_2$). We obtain
\begin{equ}\label{eq:2emescorrections}
	\begin{aligned}
		\dd p_2 &  =   \dd^0\left( \frac{\Phi_2}{p_2}\right)+ \dd^-\left(
		\frac{\Phi_2}{p_2}\right) -\dd\left( \frac{\Phi_2}{p_2}\right) \\
		&=  \frac{1}{p_2}\sum_b {p_b}\mpart{\Phi_2}{q_b}\dt -
\frac{\phi_2\Phi_2}{p_2^2}
		\dt -\dd\left( \frac{\Phi_2}{p_2}\right)\\
		&= -\frac{1}{p_2}\sum_b {p_b}w_b\dt -\frac{\phi_2\Phi_2}{p_2^2} \dt
		 -\dd\left( \frac{\Phi_2}{p_2}\right)~.
	\end{aligned}
\end{equ}
This is exactly what we found in \eref{eq:decompp2naive}. 
We deal next with 
the terms $-{p_b}w_b/p_2\,\dt$ in \eref{eq:2emescorrections}. Using
\lref{lem:remplaceravg} with $f={p_b}w_b/p_2$ (and therefore with
$Qf= {p_b}W_b/p_2^2$), we find, since  $\langle f\rangle =
 {p_b}\avg{w_b}/p_2=0$, that for $b=1,3$, 
\begin{equ}\label{eq:intermpbwb}
	\begin{aligned}
		\frac{{p_b}w_b}{p_2}\dt &= -\frac{1}{p_2^2} \left[-p_b^2w_b + (\phi_b + \tau_b
		-\gamma_bp_b){W_b}  \right]\dt \\
		& \qquad   - \frac{1}{p_2^2} {\sqrt{2\gamma_bT_b}{W_b}}\dd B^b_t +
		\frac{2}{p_2^3}{p_b}W_b\phi_2 \dt+ \dd \bigohneg{2}~,
	\end{aligned}
\end{equ}
where here and in the sequel, we denote by $\dd \bigoh{k}$
any generic expression of the kind $\dd w(X_t)$
for some function $w=\bigoh{k}$ on $\Omega$. Here $\dd \bigohneg{2} = \dd (
{p_b}W_bp_2^{-2})$. Substituting \eref{eq:intermpbwb} into \eref{eq:2emescorrections} leads to
\begin{equ}\label{eq:3emescorrections}
	\begin{aligned}
		\dd p_2& = I\dt + J\dt + \frac{1}{p_2^2}\sum_b \sqrt{2\gamma_bT_b}{W_b}\dd
B^b_t
		+ \dd \left(-\frac{\Phi_2}{p_2}
 +\bigohneg{2}\!\right) ~,
	\end{aligned}
\end{equ}
with
\begin{equ}
	\begin{aligned}
		I &=- \sum_b\frac{{p_b^2}w_b -(\phi_b + \tau_b -\gamma_bp_b){W_b}}{p_2^2}
		-\frac{\phi_2\Phi_2}{p_2^2}~,\\
		J & = \frac{2}{p_2^3}\sum_b p_bW_b{\phi_2}~.
	\end{aligned}
\end{equ}
We next deal with the terms $I\dt $ and $J\dt$. 

First,
we show that $\avg I = 0$. It is immediate that $\avg{p_2^{-2} p_b^2 w_b}$ and
$\avg{p_2^{-2} (\tau_b - \gamma_bp_b)W_b}$ are zero. Moreover, $\avg{p_2^{-2}
	\phi_2\Phi_2} = -\frac 1 2p_2^{-2}\avg{\partial_{q_2}\Phi_2^2} = 0$.  Thus,
\begin{equ}
	\begin{aligned}
		\avg I &= \sum_b\avg{\frac{1}{p_2^2}\phi_bW_b} = \sum_b\avg{\frac{w_b -
				u_b}{p_2^2} W_b} \\
		& = - \sum_b\left( \frac{ \avg{\partial_{q_2}W_b^2}  }{2p_2^2}+\frac{ u_b
			\avg{W_b} }{p_2^2}\right)= 0~.
	\end{aligned}
\end{equ}
Since $I$ is of order $-2$ and $\avg I
= 0$, we find that $QI$ is of order $-3$ and thus
$\dd^-(QI)=\bigohneg{4} \dd t$.
Applying \lref{lem:remplaceravg} with $f=I$, we find
\begin{equ}\label{eq:Idt first eq}
	I\dt  =  - \dd^0\left(QI\right)+   \bigohneg{4} \dd t+ \dd \bigohneg{3}~.
\end{equ}
Using that $\avg{QI}=0$, the definition \eref{e:ddd} of $\dd^0$ leads,
upon inspection, to
\begin{equ}
	\dd^0\left(QI\right)= \sum_b w_{b}\partial_{p_b}(QI)\dd t+  \mathcal E\dd t
	+\sum_b   \bigohneg{3}\dd B^b_t~,
\end{equ}
where $ \mathcal E$ is a sum of terms of order $-3$ and $\avg{\mathcal
  E}=0$. 
Applying  \lref{lem:remplaceravg} to $
w_{b}\partial_{p_b}(QI)\dd t$
and $\mathcal E \dt$, we obtain
\begin{equ}\label{eq:d0QI}
	\dd^0\left(QI\right)= \sum_b \avg {w_{b}\partial_{p_b}(QI)}\dd t+
	\bigohneg{4}\dd t +\sum_b   \bigohneg{3}\dd B^b_t~.
\end{equ}
Using the definition of $w_b$, integrating by parts once and using
\eref{eq:redefQf}, we have for $b=1,3$,
\begin{equ}
	\langle w_{b}\partial_{p_b}(QI)\rangle   =   \langle
	\partial_{q_2}(W_{b})Q(\partial_{p_b}I)\rangle =-\langle
	W_{b}\partial_{q_2}Q(\partial_{p_b}I)\rangle =  -\frac 1 {p_2}\langle W_{b}
	\partial_{p_b}I\rangle~.
\end{equ}
Since $\partial_{p_b}I = -p_2^{-2}\left(2{p_b}w_b
+\gamma_b{W_b}\right)$, we get
\begin{equ}\label{e:almostthereI}
	\langle w_{b}\partial_{p_b}(QI)\rangle =\avg{\frac 1 {p_2^3}W_b \left(2{p_b}w_b
		+\gamma_b{W_b}\right)} =  \frac 1 {p_2^3} \gamma_b \avg{W_b^2}~,
\end{equ}
where again we have used that $\avg{W_bw_b} = \frac 1 2
\avg{\partial_{q_2}W_b^2} = 0$. 
Substituting \eref{e:almostthereI} into \eref{eq:d0QI} and then the
result into \eref{eq:Idt first eq} we finally get
\begin{equ}\label{eq:Idtfinal}
	I\dt = -\frac \alpha  {p_2^3}\dd t+\sum_b   \bigohneg{3}\dd B^b_t+  
	\bigohneg{4} \dd t+ \dd \bigohneg{3}~,
\end{equ}
where
$$\alpha = \sum_b  \gamma_b \avg{W_b^2}~.$$

We next deal with the term $J\dt$ of \eref{eq:3emescorrections}. First, by
\lref{lem:remplaceravg}, 
\begin{equ}\label{eq:JavgJ}
	J\dt = \avg J\! \dt + \bigohneg{4}\dt + \sum_b \bigohneg{4}\dd B^b_t+\dd
\bigohneg{4}~.
\end{equ}
Unfortunately, $\avg J \neq 0$,\footnote{For example if $W_b = -\cos(q_2-q_b)$,
	there are in $\avg J$ some terms of the kind
	$\avg{p_3\cos(q_2-q_1)\sin(q_2-q_3)}$ and 
$\avg{p_1\sin(q_2-q_1)\cos(q_2-q_3)}$
	which are non-zero.} and we will need some more subtle identifications.
Integrating by parts, we have
\begin{equ}\label{eq:Jexpand}
	\begin{aligned}
		\avg J &= \frac{2p_b}{p_2^3}\sum_b \avg{W_b{\phi_2}} =
		-\frac{2p_b}{p_2^3}\sum_b \avg{W_b \partial_{q_2}{\Phi_2}}  \\
		&=  \frac{2p_b}{p_2^3}\sum_b \avg{(\partial_{q_2}W_b) \Phi_2} =
		\frac{2p_b}{p_2^3}\sum_b \avg{w_b{\Phi_2}}\\
		& =- \frac 1 {p_2^{3}} \sum_b p_b \partial_{q_b}\langle \Phi_2^2 \rangle ~.
	\end{aligned}
\end{equ}
On the other hand, since $p_2^{-3}\langle \Phi_2^2 \rangle$ does not depend on
$q_2$, we find $\dd^+(p_2^{-3}\langle \Phi_2^2 \rangle)=0$, so that
\begin{equ}\label{e:almostJ}
	\begin{aligned}
		\dd \left(\frac {\langle \Phi_2^2 \rangle}{p_2^3}\right)&= \dd^0 \left(\frac
		{\langle \Phi_2^2 \rangle}{p_2^3}\right)+\dd^-\left( \frac {\langle \Phi_2^2
			\rangle}{p_2^3}\right)\\
		& =  \sum_b p_b \partial_{q_b} \left(\frac {\langle \Phi_2^2
			\rangle}{p_2^3}\right) \dt +\bigohneg{4}\dt~.
	\end{aligned}
\end{equ}
Combining \eref{eq:Jexpand} and \eref{e:almostJ} we find
$$
 \avg J\!\dt = \bigoh{-4} \dt+ \dd(p_2^{-3}\avg \Phi_2^2) = 
\bigoh{-4}\dt+ \dd \bigoh{-3}~,
$$
 so that from \eref{eq:JavgJ} we obtain 
$$
J\dt =   \bigohneg{4}\dt + \sum_b \bigohneg{4}\dd B^b_t+ \dd \bigohneg{3}~.
$$
This together with \eref{eq:3emescorrections} and \eref{eq:Idtfinal} finally
shows that
\begin{equ}\Label{eq:5emescorrections}
       \dd p_2 = - \left(\frac \alpha  {p_2^3}   + \bigohneg{4}\!\right)
\dt+\sum_b
       \left(\frac{\sqrt{2\gamma_bT_b}{W_b}}{ p_2^2}+\bigohneg{3}\!\right)\dd
B^b_t +
       \dd \left(-\frac{\Phi_2}{p_2}+\bigohneg{2}\!\right) ~,
\end{equ}
which implies \eref{e:p1nm4equiv} and completes the proof of \pref{p:mainprop}.

\begin{remark}\label{rem:langevin} We can argue (in a nonrigorous way) that when $|p_2|$ is very
large, the dynamics of $\tilde p_2$ is approximately that of a particle
interacting with two ``effective'' heat baths at temperatures $T_1$ and $T_3$,
but with some coupling of magnitude $p_2^{-4}$. Indeed, we can write
\eref{e:p1nm4equiv} in the canonical ``Langevin'' form
\begin{equ}
 \dd \tilde  p_2(t) =  \sum_b \big(-\tilde \gamma_b(X_t) \tilde  p_2(t)\dd t+
 \sigma_b(X_t) \dd B^b_t\big)~,
\end{equ}
with $\sigma_b(q, p) = {\sqrt{2\gamma_bT_b}{W_b}}/{p_2^2}+\bigoh{-3}$ as in
\pref{p:mainprop} and $\tilde \gamma_b(q, p) =    { {\gamma_b 
\avg{W_b^2}}}/{p_2^4} + \bigoh{-5}$. We would like to introduce an effective
temperature $\tilde T_b$ by some Einstein-Smoluchowski relation of the kind
$\sigma_b^2/(2\tilde \gamma_b) = \tilde  T_b $ in the limit $|p_2|\to \infty$.
Unfortunately, 
$$
\lim_{|p_2|\to\infty}\frac{\sigma_b^2}{2\tilde \gamma_b} = \frac
{W_b^2}{\avg{W_b^2}}T_b~,
$$
which instead of a constant is an oscillatory quantity (with mean $T_b$). Now
observe that these oscillations disappear if we approximate the oscillatory term
$W_b$ in $\sigma_b$ by its quadratic mean $\avg{W_b^2}\!\vphantom{0}^{1/2}$.
This approximation is reasonable in the following sense: for small $t$ and large
$|p_2|$, we have that $p_2(s) \approx p_2(0)$ for $s\leq t$, so that 
$$
\int_0^t  \frac{\sqrt{2\gamma_b T_b}W_b}{{p_2^2(s)}} \dd B^b_s  \approx 
\frac{\sqrt{2\gamma_b T_b}}{{p_2^2(0)}}
{\avg{W_b^2}\!\vphantom{0}^{1/2}}M(t)\quad \text{with } 
M(t) =\int_0^t \frac{W_b}{\avg{W_b^2}\!\vphantom{0}^{1/2}}\dd B^b_s~.
$$
But then, by the Dambis-Dubins-Schwarz representation theorem, there is another
Brownian motion
$\tilde B^b$ such that $M(t) = \tilde B^b_{\tau(t)}$ with $\tau(t) =  \int_0^t
{W_b^2}/{{\avg{W_b^2}}}\dd s$.
Clearly, when $|p_2|$ is very large, $\tau(t) \approx t$ so that $M(t)$ is very
close to $\tilde B^b_{t}$.
In this sense, when $|p_2|\to \infty$, it is reasonable to approximate
$({\sqrt{2\gamma_b T_b}}W_b/{{p_2^2}}) \dd B^b_s$
with $({\sqrt{2\gamma_b T_b}}\avg{W_b^2}\!\vphantom{0}^{1/2}/p_2^2) \dd \tilde
B^b_s$, so that the Einstein-Smoluchowski
relation indeed holds with effective temperature $\tilde T_b = T_b$.
\end{remark}

\begin{remark} 
The ergodicity of 1D Langevin processes is well understood:
for any $\delta \in (-1,0)$, processes satisfying an SDE of the kind
$$\dd p \sim - C_1p^\gd \dt  + C_2\dd B_t$$
asymptotically (when $|p|\to \infty$)
are typically ergodic with a rate bounded above and below by $\exp(-c_\pm t^{(1+\gd)/(1-\gd)})$ for
some constants $c_+, c_->0$ (see
\cite{douc_subgeometric_2009,hairer_how_2009} and references therein,
in particular \cite{hairer_how_2009} for the lower bound). As argued in \rref{rem:langevin}, the variable $\tilde p_2$ (which is expected
to be the component of the system that limits the convergence rate) essentially
obeys an equation of the kind 
$\dd p  \sim -C_1p^{-3}\dt + C_2 p^{-2}\dd B_t$
asymptotically. It is easy to check that a change of
variable $y = p^{3}$ yields the asymptotic 
dynamics $\dd y \sim - C_1' y^{-1/3}\dt + C_2'\dd B_t$
so that with $\delta = -1/3$, we expect a rate $\exp(-c t^{1/2})$.
This suggests that the rate of convergence we find is optimal.
\end{remark}

\section{Lyapunov function}\label{sec:effectivedyntoLyapunov}
\label{sec:Lyapunovconstr}

We now prove \pref{prop:Lyapunov}. Throughout this section, $\tilde p_2$
is the function defined in \pref{p:mainprop}.
The basic idea is to consider a Lyapunov function
$$
V \sim \rho(p)\p2t2 e^{\frac \beta 2 \p2t2} + e^{\beta H}~,
$$
where $\rho(p)$ is non-zero only when $|p_2|$ is much larger than $|p_1|$ and $|p_3|$.
We will obtain that
$L V \lesssim -\phi(V)$, with $\phi(s) \sim  {s}/\log(s)$ as in 
\pref{prop:Lyapunov}. The fact that we do {\em not} have a bound
of the kind $L V \lesssim -cV$ (which would yield exponential ergodicity)
comes from the very slow decay of $p_2$. The basic idea is that, when $p_2\to \infty$ and $p_1, p_3\sim 0$,
$$L\tilde p_2 \sim - p_2^{-3}~,\quad \text{ so that }\quad L\Big(\p2t2 e^{{\frac \beta 2} \p2t2}\Big)
\sim -e^{{\frac \beta 2} \p2t2} \sim - \frac V{p_2^2} \sim -\frac V{\log V}~.$$

We now introduce the necessary tools to make this observation rigorous.

\begin{lemma}\label{l:LebetaH} For $\beta  > 0$ small enough, there are
constants $C_1, C_2>0$ such that
\begin{equ}\Label{eq:LbetaHseul}
Le^{\beta H} \leq  (C_1 - C_2(p_1^2 + p_3^2))e^{\beta H}~.
\end{equ}
\end{lemma}
\begin{proof} We have $
Le^{\beta H} = \sum_b  \left(-\gamma_b\beta (1-\beta T_b) p_b^{2}+ \beta\tau_b p_b +
\gamma_b\beta T_b\right)e^{\beta H}
$. If $\beta < 1/\max(T_1, T_3)$,
then $\gamma_b\beta (1-\beta T_b) > 0$. Moreover, since $\beta \tau_b p_b < \frac 12
{\gamma_b\beta (1-\beta T_b)}  p_b^2 + C$ for $C$ large enough,
we find the desired bound.
\end{proof}

\begin{lemma}\label{l:Lelambdap2tilde}
For $\beta  > 0$ small enough, there is a constant $C_3 > 0$ such that on
$\Omega^\dagger=\{x\in \Omega:\, p_2\neq 0\}$,
\begin{equ}[eq:Lfp2tilde]
L \big(\p2t2 e^{\frac \beta 2  \p2t2}\big) \leq  \big(-C_3  +
\bigohneg{1}\!\big)e^{\frac \beta 2  \p2t2}~.
\end{equ}
\end{lemma}
\begin{proof}Introducing $f(s) = s^2e^{\frac \beta 2  s^2}$, we have by the
It\=o formula and \pref{p:mainprop} that
\begin{equs}
\dd \big(\p2t2 e^{\frac \beta 2  \p2t2}\big) & = \dd f(\tilde p_2)
= f'(\tilde p_2) ( a\dt + \sum_b  \sigma_b\dd B^b_t) + \frac 12 f''(\tilde p_2)
\sum_b \sigma_b^2  \dt\\
& = (2 \tilde p_2+\beta  \tilde p_2^{\,3})e^{\frac \beta 2 \p2t2} ( a\dt +
\sum_b  \sigma_b\dd B^b_t) + \frac 12 (2+5\beta  \p2t2 +\beta^2 \tilde
p_2^{\,4})e^{\frac \beta 2  \p2t2} \sum_b \sigma_b^2  \dt~.
\end{equs}
Now since $a =- \alpha{p_2^{-3}} + \bigohneg{4}$ with
$\ga= \sum_b \gamma_b \avg{ W_b^2}$, $\sigma_b
={\sqrt{2\gamma_bT_b}{W_b}}{p_2^{-2}}+\bigohneg{3}$, and $\tilde p_2^{\,k} =
p_2^k + \bigoh{k-1}$ for all $k$, we find after taking the expectation value
\begin{equs}
L \big(\p2t2 e^{\frac \beta 2  \p2t2}\big) &= \big(-\alpha \beta +
 {\beta^2} \sum_b\gamma_b T_b W_b^2  + \bigohneg{1}\!\big)e^{\frac \beta 2  \p2t2}~,
\end{equs}
which gives the desired bound if $\beta$ is small enough (recall that the $W_b^2$ are bounded).
\end{proof}

\noindent{\bf Convention}: We fix $\beta>0$ small enough so that the conclusions of
\lref{l:LebetaH} and \lref{l:Lelambdap2tilde} hold.\medskip 

Let $k\geq 1$ be an integer and $R>0$ be a constant (which we will fix later). We
split $\Omega$ into three disjoint sets $\Omega_1, \Omega_2, \Omega_3$ defined
by
\begin{myitem}
\item $\Omega_1 = \{x\in \Omega : |p_2| < (p_1^2+p_3^2)^k + R\}$,
\item $\Omega_2 = \{x\in \Omega : (p_1^2+p_3^2)^k + R \leq |p_2| \leq 2
(p_1^2+p_3^2)^k +2R\}$,
\item $\Omega_3 = \{x\in \Omega : |p_2| > 2 (p_1^2+p_3^2)^k + 2R\}$.
\end{myitem}

Fix some $m, n\in \mathbb N$ and $\ell \geq 1$. On $\Omega_2\cup \Omega_3$,
we have by definition $|p_2| \geq  (p_1^2+p_3^2)^k + R$, so that
$$
\left|\frac{p_1^{n}p_3^{m}}{p_2^{\ell}}\right| \leq \frac{ |p_1^np_3^m|}
{((p_1^2+p_3^2)^k + R)^{\ell}} \qquad \text{(on $\Omega_2\cup\Omega_3$)}~.
$$
Clearly, if $k$ and $R$ are large enough, the right-hand side is bounded by an
arbitrarily small constant.
Therefore, any given $\bigoh{-1}$ is also bounded by an arbitrarily small constant on $\Omega_2\cup\Omega_3$ 
provided that $k$ and $R$ are large enough, 
since it is
by definition a sum {\em finitely} many terms of order less or equal to -1. Using this, we obtain

\begin{lemma}\label{l:omegas}For $k$ and $R$ large enough, there are
constants $C_4,\dots,  C_7>0$ such that the
following properties hold on $\Omega_2\cup\Omega_3$:
\begin{equ}\label{eq:lomegasfirst}
|\p2t2 - p_2^2| < C_4~,
\end{equ}
\begin{equ}\label{eq:lomegassec}
L \big(\p2t2 e^{\frac \beta 2  \p2t2}\big) \leq  -C_5e^{\frac
\beta 2 \p2t2}~,
\end{equ}
\begin{equ}\label{eq:lomegasthird}
 C_{6} e^{-\frac \beta 2 
\left(p_1^2 +p_3^2\right)}e^{ \beta H}
\leq e^{\frac \beta 2 \p2t2} \leq C_7 e^{-\frac \beta 2 
\left(p_1^2 +p_3^2\right)}e^{ \beta H}~.
\end{equ}
\end{lemma}
\begin{proof}
Since
$\tilde p_2 = p_2 +  {\Phi_2(q)}/{p_2} + \bigoh{-2}$, we have $\p2t2 =
p_2^2 + 2 {\Phi_2(q)} + \bigoh{-1}$. By taking $k$ large enough, 
the $\bigoh{-1}$ here is bounded by a constant
on the set $\Omega_2\cup \Omega_3$, which implies \eref{eq:lomegasfirst}.
Moreover, for large $k$ and $R$, the $\bigohneg{1}$ in \eref{eq:Lfp2tilde}
is also bounded on $\Omega_2\cup \Omega_3$ by an arbitrarily small constant,
 which implies \eref{eq:lomegassec}.
To prove \eref{eq:lomegasthird}, observe that
$$e^{\frac \beta 2 \p2t2} = e^{\frac \beta 2 
\left(\p2t2- p_2^2 - p_1^2 -p_3^2 - U(q)\right)}e^{ \beta H}~,$$
where $U(q)$ contains all the potentials appearing in $H$. This together with
the boundedness of $U$
and \eref{eq:lomegasfirst}, implies \eref{eq:lomegasthird}.
\end{proof}

\noindent{\bf Convention}: We fix $k$ and $R$ such that the conclusions of \lref{l:omegas} hold.

\begin{definition}
Let $\chi: \mathbb R \to [0,1]$ be a smooth function such that $\chi(s) = 0$
when $|s| < 1$ and $\chi(s) = 1$ when $|s|>2$.
We introduce the cutoff function 
$$
\rho(p) = \chi\left(\frac{p_2}{(p_1^2+p_3^2)^k + R}\right)~,
$$
and the Lyapunov function
$$
V = 1 + A \rho(p) \p2t2 e^{\frac \beta 2  \p2t2} + e^{\beta H}~,
$$
with $A>0$ (to be chosen later).
\end{definition}

By construction $\rho(p)$ is smooth, $\rho(p) = 0$ on $\Omega_1$ and $\rho(p)
=1$ on $\Omega_3$, with some transition on $\Omega_2$.
Remember that $\tilde p_2$ is by construction smooth on
$\Omega^\dagger$, \ie when $p_2\neq 0$. In particular, 
since $\Omega_2\cup\Omega_3 \subset \Omega^\dagger$, the function
$\rho(p) \p2t2 e^{\frac \beta 2  \p2t2}$ is smooth on $\Omega$,
and so is $V$. We can now finally give the

\begin{proof}[Proof of \pref{prop:Lyapunov}]
We show here that $V$ satisfies the conditions enumerated
in \pref{prop:Lyapunov} if $A$ is large enough. Let us first 
prove the first statement, which is that there exist $c_1, c_2>0$ such that
\begin{equ}\label{eq:ineqvn2}
	1 + c_1 e^{\beta H}\leq V\leq c_2(1+p_2^2) e^{\beta H}~.
\end{equ}
Clearly the lower bound on $V$ holds.
We now prove the upper bound. Throughout the proof, we
denote by $c$ a generic
positive constant which can be each time different. 
Since $\rho \neq 0$ only on $\Omega_2 \cup
\Omega_3$, 
we have by \eref{eq:lomegasfirst} and \eref{eq:lomegasthird},
\begin{equs}
|A \rho(p)\p2t2 e^{\frac \beta 2  \p2t2}| & \leq c (p_2 +
C_4)^2e^{-\frac \beta 2 
\left(p_1^2 +p_3^2\right)}e^{ \beta H}\\
& \leq c(p_2^2+ 2C_4 p_2 + C_4^2)e^{ \beta H} \leq c(1+p_2^2)e^{ \beta H}~.
\end{equs}
But then $V \leq 1 + c(1+p_2^2) e^{\beta H} \leq c(1+p_2^2)e^{\beta H}$, where
the last inequality holds because $H$ is bounded below, so that
$e^{\beta H}$ is bounded away from zero.

Let us now move to the second statement of \pref{prop:Lyapunov},
which is that for $c_3, c_4$ large enough and a compact set $K$,
\begin{equ}\label{eq:ineqLvnrappel}
	LV \leq c_3 \ind_K -\phi(V) \quad \text{with } \phi(s) = \frac
{c_4\,s}{2+\log(s)}~.
\end{equ}
We first show that
\begin{equ}\label{eq:prooflyapunovaim}
LV \leq c\ind_K- ce^{\beta H}\qquad \text{with}\quad K =
\{x\in \Omega_1\cup \Omega_2 : p_1^2 + p_3^2 \leq M\}~,
\end{equ}
for some large enough $M$. Clearly $K$ is compact,
since $\Omega_1\cup \Omega_2  = \{x\in \Omega : |p_2| \leq 2 (p_1^2+p_3^2)^k +2R\}$.
\begin{myitem}
\item On $\Omega_1$ we simply have $V = 1 + e^{\beta H}$. By
\lref{l:LebetaH}, we have $LV \leq  (C_1 - C_2 (p_1^2 + p_3^2))e^{\beta H}$.
Since $\Omega_1 \setminus K = \{x\in \Omega_1 : p_1^2 + p_3^2 > M\}$, we have for large enough $M$ that $LV \leq -ce^{\beta H}$ on $\Omega_1 \setminus K$, and therefore
\eref{eq:prooflyapunovaim} holds on $\Omega_1$.

\item On $\Omega_2$, the key is to observe that there is a polynomial $z(p_1,
p_2, p_3)$
such that 
$$
|L (A\rho(p) \p2t2 e^{\frac \beta 2  \p2t2})| \leq z(p) e^{\frac
\beta 2  \p2t2} \leq 
C_7 z(p) e^{-\frac \beta 2
\left(p_1^2 + p_3^2\right)}e^{\beta H}~,
$$ 
where the second inequality comes from \eref{eq:lomegasthird}.
Now,  since $p_1^2+p_2^2 \sim |p_2|^{1/k}$ on $\Omega_2$, we have that
$z(p) e^{-\frac \beta 2 
\left(p_1^2 +p_3^2\right)}$ is bounded on $\Omega_2$. 
Therefore, by this and \lref{l:LebetaH}, we have on $\Omega_2$,
\begin{equs}\Label{eq:LVdansOmega2}
LV &\leq  \left( C_7 z(p) e^{\frac \beta 2 
\left(- p_1^2 -p_3^2\right)}+  C_1 - C_2 (p_1^2 + p_3^2)\right)e^{\beta H}\\
& \leq  \left( c - C_2 (p_1^2 + p_3^2)\right)e^{\beta H} ~.
\end{equs}
which, as in the previous case, implies that
\eref{eq:prooflyapunovaim} holds on $\Omega_2$ if $M$ is large enough.

\item On $\Omega_3$, which is the critical region, we have
 $V = 1 + A \p2t2 e^{\frac \beta 2  \p2t2} +e^{\beta H}$.
By \lref{l:LebetaH} and \eref{eq:lomegassec}, it holds in $\Omega_3$ that
\begin{equ}\label{eq:LVdansOmega3}
LV \leq  (C_1 - C_2 (p_1^2 + p_3^2))e^{\beta H} 
 -C_5 A e^{\frac \beta 2 \p2t2} ~.
\end{equ}
On the set $\{x\in \Omega_3 : C_1 - C_2 (p_1^2 + p_3^2) \leq -1\}$,
we simply have $LV \leq  -e^{\beta H}$, so that \eref{eq:prooflyapunovaim} holds
trivially.
On the other hand, on the set 
$\{x\in \Omega_3 : C_1 - C_2 (p_1^2 + p_3^2) > -1\}$ the quantity
$p_1^2+p_3^2$ is bounded, so that 
$e^{\frac \beta 2 \p2t2} \geq c e^{\beta H}$  by \eref{eq:lomegasthird},
which with \eref{eq:LVdansOmega3} implies that
$$
LV \leq  (C_1 - C_2 (p_1^2 + p_3^2))e^{\beta H} 
 -c A e^{\beta H} \leq  (C_1 - cA)e^{\beta H}~.
$$
By making $A$ large enough, we again find a bound  $LV \leq  -ce^{\beta H}$,
so that \eref{eq:prooflyapunovaim} holds.
\end{myitem}

Therefore, \eref{eq:prooflyapunovaim} holds on all of $\Omega$.
To obtain \eref{eq:ineqLvnrappel}, we need only show that
$e^{\beta H} \geq cV/(2+\log V)$. By
the boundedness of the potentials and the definition
of $V$, we have
$1+ p_2^2 \leq 2 H + c \leq c \log(e^{\beta H}) + c  \leq c\log V +c  \leq
c(\log V + 2)$.
But then by \eref{eq:ineqvn2}
we indeed have that $e^{\beta H} \geq  c V/(1+p_2^2) \geq cV/(2+\log V)$. This
completes the proof of \pref{prop:Lyapunov}.
\end{proof}

\begin{remark} 
 The external forces $\tau_b$ and the
pinning potentials $U_i$ (if non-zero) do not play
a central role in the properties of the Lyapunov function.
On the contrary, the interaction potentials $W_b$ are very
important, since we need $\alpha = \sum_b  \gamma_b \avg{W_b^2}$ to be strictly
positive.
\end{remark}

\begin{remark} 
 Although we assume throughout that $T_1$ and $T_3$ are strictly positive, the
computations that lead to the Lyapunov function
apply to zero temperatures as well (the
temperatures only appear in some non-dominant terms in $V$ and $L V$).
In that case, the existence of an invariant measure
can still be obtained by compactness arguments 
(see \eg  Proposition 5.1 of \cite{hairer_slow_2009}).
However, the smoothness, uniqueness and
convergence assertions do not necessarily hold: when $T_1=T_3 = 0$ the system is
deterministic, the
transition probabilities are not smooth, and there is at least one invariant
measure concentrated at each stationary point of the system.
The positive temperatures assumption is
crucial in the next section.
\end{remark}

\section{Smoothness and irreducibility}
\label{sec:controlirred}

This section is devoted to proving that the hypotheses of \tref{th:douc} other
than the existence of the Lyapunov function are satisfied. More precisely we
will
prove the following proposition.

\begin{proposition}\label{prop:controlstart}
The following properties hold.
\begin{myenum}
	\item[(i)] The transition probabilities $P^t(x, \argcdot)$ have a density
	$p_t(x, y)$ that is smooth in $(t, x, y)$ when $t>0$. In particular, the
process
	is strong Feller.
	\item[(ii)] The time-1 skeleton $(X_{n})_{n=0, 1, 2, \cdots}$ is irreducible,
	and the Lebesgue measure $m$ on $(\Omega, \mathcal B)$ is a maximal irreducibility measure.
	\item[(iii)] Every compact set is petite.
\end{myenum}
\end{proposition}

In a sense, (i) shows that we have some effective diffusion in all directions at
very short times, and (ii) shows that every part of the phase space is
eventually reached with positive probability. Observe that (iii) follows from
(i) and (ii). Indeed, by (i), (ii) and Proposition 6.2.8 of
\cite{meyn_markov_2009}, every compact set is petite for the time-1 skeleton.
But then every compact set is also petite with respect to the process $X_t$
(simply by choosing a sampling measure on $[0, \infty)$ that is
concentrated on $\mathbb N$). Therefore,
we need only prove (i) and (ii), which we do in the next two subsections.

\subsection{Smoothness}
\label{sec:smooth}

We show here that the semigroup has a smoothing effect. More
specifically, we show that a H\"ormander bracket condition is satisfied, so that
the transition probability $P^t(x, \dd y)$ has a density $p_t(x, y)$ that is
smooth in $t,x$ and $y$, and every invariant measure has a smooth density
\cite{hormander_1967}.

We identify vector fields over $\Omega$ and the corresponding
first-order differential operators in the usual way (we identify the tangent
space of $\Omega$ with $\mathbb R^6$). This enables us to
consider Lie algebras of vector fields over $\Omega$ of the kind $\sum_i (f_i(q,
p)\partial_{q_i} + g_i(q,p)\partial_{p_i})$, where the Lie bracket
$[\argcdot,\argcdot]$ is the usual commutator of two operators.

\begin{definition}We define $\mathcal M$ as the smallest
Lie algebra that
\begin{myitem}
	\item[(i)] contains the constant vector fields $\partial_{p_1},
\partial_{p_3}$,
	\item[(ii)] is closed under the operation $[\argcdot, A_0]$, where 
	$$A_0 = \sum_{i=1}^3 \left(p_i \partial_{q_i} - u_i\partial_{p_i} \right) +
	\sum_{b}(w_b(\partial_{p_b}-\partial_{p_2}) + \tau_b \partial_{p_b} -
	\gamma_b p_b \partial_{p_b}) 
	$$
	is the drift part of $L$.
\end{myitem}
\end{definition}

By the definition
of a Lie algebra, $\mathcal M$ is closed under
linear combinations and Lie brackets.

\begin{lemma} H\"ormander's bracket condition is satisfied. More precisely, for
all $x = ( q,  p)$, the set $\{v(x): v\in \mathcal M \}$ spans
$\mathbb R^6$.
\end{lemma}
\begin{proof}
By definition, the constant vector fields $\partial_{p_1}$ and $\partial_{p_3}$
belong to $\mathcal M$. Moreover, for $b=1,3$,
$[\partial_{p_b}, A_0] = \partial_{q_b} - \gamma_b \partial_{p_b}$.
Since $\mathcal M$ is closed under linear combinations and 
$\partial_{p_b}\in \mathcal M$, it follows that $\partial_{q_b} \in \mathcal M$
for $b=1,3$. Thus it only remains to show that at each $x\in \Omega$, we can
span the directions of $\partial_{q_2}$ and $\partial_{p_2}$. In the following, 
$f$ denotes a generic function on $\Omega$ that can be each
time different. We have $
[\partial_{q_b}, A_0] = w'_b(q_2-q_b)\partial_{p_2} +f( q) \partial_{p_b}$
so that commuting $n-1$ times with $\partial_{q_b}$ we get that for all $n\geq
1$
\begin{equ}\label{eq:wndansM}
w^{(n)}_b(q_2-q_b)\partial_{p_2}+f( q) \partial_{p_b} \in \mathcal M~.
\end{equ}
Commuting the above  with $A_0$, we find that for  all $n\geq 1$,
\begin{equ}\label{eq:wndqdansM}
w^{(n)}_b(q_2-q_b)\partial_{q_2}+f( q,  p)\partial_{p_2}+f( q) \partial_{p_b}+f(
q, p) \partial_{q_b} \in \mathcal M~.
\end{equ}
By \aref{as:assumptioncoupling}, there is some $b\in \{1,3\}$ such that 
for any fixed $x\in \Omega$, there is an integer $n\geq 1$ such that
$w^{(n)}_b(q_2-q_b)\neq 0$.
Thus, by \eref{eq:wndansM} and \eref{eq:wndqdansM} the proof is complete.
\end{proof}

Thus, we have proved \pref{prop:controlstart} (i).

\subsection{Irreducibility}

We show in this section that the process has an irreducible skeleton. We
give in fact two different proofs. The first one
is given in a general and abstract framework, and works for chains of any
lengths. The second one is more explicit, gives more than the irreducibility 
of a skeleton, but relies strongly on the fact that
the chain is made of only three rotors.

\subsubsection{Abstract version}

Consider the transition probabilities $\tilde  P^t(\argcdot, \argcdot)$ of the system
at equilibrium, \ie with parameters $\tau_1 = \tau_3 = 0$ and $T_1 = T_3 =T$ for some $T>0$. For all $x$ and $t$, the measures $P^t(x,
\argcdot)$ and $\tilde P^t(x, \argcdot)$ are equivalent. This equivalence holds
because any change of the parameters $\tau_1, \tau_3$ (respectively $T_1,T_3$)
can be absorbed by shifting (respectively scaling) the Brownian motions
appropriately. Therefore, it is enough to prove the irreducibility claim
at equilibrium.

At equilibrium, the Gibbs measure $\nu$ with density $\frac 1Z\exp(- H/T)$ is invariant
(with some normalisation constant $Z$) as mentioned earlier. Note that we do not
assume {\it{a priori}} that $\nu$ is the unique invariant measure at
equilibrium, nor
that the system at equilibrium is irreducible. The only two properties that we
need are invariance and  (everywhere) positiveness of the density of $\nu$ .

\begin{lemma}\label{lem:detailedbalance}The equilibrium transition probabilities
satisfy the following property: for every measurable set
$S$ one has for all $t$
$$
\int_S \tilde{P}^t(x,S^c)\dd \nu = \int_{S^c} \tilde{P}^t(x, S)\dd \nu ~. 
$$
\end{lemma}
\begin{proof}
We have by the invariance of $\nu$,
\begin{equ}
	\begin{aligned}
	\int_{S^c} &\tilde{P}^t(x, S)\dd \nu -   \int_{S} \tilde{P}^t(x, S^c)\dd \nu =
	\int_{S^c} \tilde{P}^t(x, S)\dd \nu +   \int_{S} (\tilde{P}^t(x, S) -1)\dd\nu\\
	&  = \int_\Omega \tilde{P}^t(x, S)\dd \nu - \int_S 1 \dd \nu  = \nu(S)-\nu(S)=
	0~,
	\end{aligned}
\end{equ}
which completes the proof.
\end{proof}

\begin{lemma}\label{lem:closedallornothing}Let $A$ be a closed set. If
$A$ is invariant under $\tilde{P}^1$ (\ie $\tilde{P}^1(x, A) = 1$ for all $x\in
A$),
then either $A = \emptyset$ or $A=\Omega$.
\end{lemma}
\begin{proof}By \lref{lem:detailedbalance}, $
\int_{ A^c} \tilde{P}^1(x,  A)\dd \nu  = \int_{ A} \tilde{P}^1(x, A^c)\dd \nu =
0$ since  $\tilde{P}^1(x, A^c) = 0$ for all $x\in A$. This implies that
$\tilde{P}^1(x,  A) = 0$ for all $x\in A^c$, since $x\mapsto
\tilde{P}^1(x,  A)$ is continuous on the open set $ A^c$ and $\nu$ has an
everywhere positive density. 
But then $\tilde{P}^t(x,  A)$  is $1$ when $x\in  A$ and $0$ when $x\in 
A^c$, so that by continuity we have $\partial A = \emptyset$. Since $\Omega$ is
connected, the conclusion follows.
\end{proof}

Note that same does not hold for non-closed sets: for example $\Omega$ minus any
set of zero Lebesgue measure is still an invariant set.

\begin{lemma} The time-1 skeleton $(X_{n})_{n=0, 1, 2, \cdots}$ is irreducible,
and the Lebesgue measure $m$ is a maximal irreducibility measure.
\end{lemma}
\begin{proof} As discussed above, it is enough to
prove the result at equilibrium,  \ie with $\tilde P^1(\argcdot, \argcdot)$.
Let $B$ be a set such that $m(B) > 0$. 
We need to show that the set $A = \{x\in
\Omega : \sum_{n=1}^\infty  \tilde{P}^n(x, B)= 0 \}$ is empty. By the smoothness
of $x\mapsto \tilde{P}^n(x, B)$, it is easy to see that $A^c = \{x\in \Omega
: \exists n>0, \tilde{P}^n(x, B) > 0  \}$ is open, so that $A$ is closed.
Moreover, for all $x\in A$ it holds that $0 = \sum_{n=1}^\infty \tilde{P}^n(x,
B) \geq \sum_{n=1}^\infty \tilde{P}^{n+1}(x, B) = \int_{\Omega } \tilde{P}^1(x,
dy)\sum_{n=1}^\infty \tilde{P}^n(y, B)$. But since by the definition of $A$ we
have
$\sum_{n=1}^\infty \tilde{P}^n(y, B) > 0$ for all $y\in A^c$, we must have
$\tilde{P}^1(x, A^c) = 0$ for all $x\in A$, so that $A$ is invariant. But then
by
\lref{lem:closedallornothing} either $A=\emptyset$ or $A = \Omega$. We need to
eliminate the second possibility. Since $m(B) > 0$ and $\nu$ has positive
density, we 
have $\nu(B)>0$. By the invariance of $\nu$, we have $\int_\Omega \tilde{P}^1(x,
B) \dd \nu  = \nu(B) > 0$. But then there is some $x\in \Omega$ such that
$\tilde{P}^1(x, B) > 0$, so that $x\in A^c$. Therefore $A\neq \Omega$, and thus
$A = \emptyset$ and the process is irreducible with measure $m$. That $m$
is a maximal irreducibility measure follows immediately from the fact that
the transition probabilities are absolutely continuous with respect to $m$.
This completes the proof.
\end{proof}

Thus, we have proved \pref{prop:controlstart}(ii), so that the proof of
\pref{prop:controlstart} is complete.

\subsubsection{Direct control version}

We give now an alternate proof of \pref{prop:controlstart}(ii).
We establish
the irreducibility of our process by using controllability arguments.
We aim to establish the controllability of \eref{eq:sde}, where the
Brownian motions $B^1_t$ and $B^3_t$ are replaced with some deterministic,
smooth controls $f_b :\mathbb R^+ \to \mathbb R$.
By absorbing some terms into the controls $f_b$, this problem is obviously
equivalent to controlling the differential equation
\begin{equ}\label{eq:controlde2}
\begin{aligned}
	\dot q_i(t)&=p_i(t)~,\\
	\dot  p_2(t)&=-\sum_{b}w_b\big(q_2(t)-q_b(t)\big) ~,\\
	\dot  p_b(t)&=f_b(t)~.
\end{aligned}
\end{equ}
In \cite{eckmann_entropy_1999} the irreducibility of chains oscillators
has been studied. The authors have proved that chains of any length
are controllable in arbitrarily small times.
This is of course not the case in our model: since
the force applied to $p_2$ is bounded by some constants
$$
K^- = \sum_{b}\min_{s\in \mathbb T} w_b(s), \quad K^+ = \sum_{b}\max_{s\in
\mathbb T} w_b(s)~,
$$
the minimal time we need to bring the system from $x^i = (q^i, p^i)$ to
$x^f = (q^f, p^f)$ is at best
proportional to $|p_2^f - p_2^i|$. On the other hand,
$q_1, p_1, q_3, p_3$ can be put into any position in arbitrarily short
time. Observe that due to \aref{as:assumptioncoupling} and the fact that
$\avg{w_b} = 0$, we have $K^- < 0 <K^+$. 
We will
prove the following proposition (remember that the positions $q_i$ 
are defined modulo $2\pi$).

\begin{proposition}\label{prop:control}The system \eref{eq:controlde2} is
approximately controllable in the sense that for all
$x^i = (q^i, p^i)$, $x^f = (q^f, p^f)$ and all
$\varepsilon>0$, there is a time $T^* > 0$ satisfying
$T^*\leq c_1+c_2|p_2^f-p_2^i|$ for some constants $c_1$ and $c_2$ such that for
all $T>T^*$  there are some smooth controls $f_1, f_3:[0,T]\to \mathbb R$ such
that the solution of \eref{eq:controlde2} with initial condition
$x^i$  satisfies $\|x(T)-x^f\| <
\varepsilon$. 
\end{proposition}

This property implies the irreducibility of the chain, since the classical
result of Stroock and Varadhan \cite{stroock_1972} links the support of the
semigroup $P^t$ and
the accessible points for \eref{eq:controlde2}, and implies in particular that
for all $x^i=(q^i, p^i)$ and $t>c_1$ the subspace
$\{x\in \Omega:\, |p_2-p_2^i|\leq (t-c_1)/c_2\}$ is included in
the support of $P^t(x^i,\argcdot)$.

The idea is the following: in the next lemma, we show how the middle
rotor can be forced into any configuration by applying
some piecewise constant force $g(t)$ to it, with $g(t) \in [K^-, K^+]$.
Then, we will argue that one can move $q_1$ and $q_3$ (on which
we have good control) in such a way that the force exerted on the middle rotor
is almost $g(t)$.

\begin{lemma}\Label{lem:controlcentre}Consider the system
\begin{equ}\label{eq:controlde3}
	\begin{aligned}
		\dot {\bar q}_2(t)&=\bar p_2(t)~,\\
		\dot {\bar p}_2(t)&=g(t) - u_2(\bar q_2(t))~,\\
	\end{aligned}
\end{equ}
and fix some initial and terminal conditions $(q_2^i, p_2^i)$ and $(q_2^f, p_2^f)$.
We claim that there is a $T^*$ satisfying
$T^*\leq c_1+c_2|p_2^f-p_2^i|$ for some constants $c_1$ and $c_2$ such that for
all $T>T^*$ there is a piecewise constant control $g(t) : \mathbb R^+ \to [K^-,
K^+]$ (with finitely many constant pieces) such that the solution of
\eref{eq:controlde3} with initial data $(q_2^i, p_2^i)$ satisfies $\bar p_2(T) =
p_2^f$ and $\bar q_2(T) = q_2^f$.
\end{lemma}
\begin{proof} We prove this result only in the case $u_2 \equiv 0$. If $p_2^f
\geq 
p_2^i$, then let $\Theta =  (p_2^f-p_2^i)/{K^+}$ and let $g(t) = K^+$ for
all $t\in [0, \Theta)$. If $p_2^f < p_2^i$ let $\Theta = (p_2^f-p_2^i)/
{K^-}$ and let $g(t) = K^-$ for all $t\in [0, \Theta)$. In both cases,
$\bar p_2(\Theta) = p_2^f$, while $\bar q_2(\Theta)$ might be anything.
Let now $K^* = \min(|K^+|, |K^-|)$, and consider some $\Delta>0$ and $ a\in
[0,K^*]$.
Assume that $g(t) =a$ when $t\in [\Theta,
\Theta+\Delta)$ and $g(t) = -a$ when $t\in [\Theta+\Delta, \Theta + 2\Delta]$.
Clearly $\bar p_2(\Theta + 2\Delta) = \bar p_2(\Theta) = p_2^f$ and
$$
\bar q_2(\Theta + 2\Delta) = \bar q_2(\Theta) + 2\Delta p_2^f + a\Delta^2~.
$$
Observe that as soon as $\Delta > \sqrt{2\pi/K^*}$, we can choose $a\in [0,
K^*]$ so that $\bar q_2(\Theta+2\Delta)$ takes any value (modulo $2\pi$). In
particular,
we can choose it to be $q_2^f$, so that we have the advertised result with $T^*
= \Theta + \sqrt{2\pi /K^*}$.
\end{proof}
\begin{remark}We have given a proof only if $u_2\equiv 0$. However, the result
remains true even if $u_2 \neq 0$, although the proof is much more involved.
Typically, if the pinning is stronger than the interaction forces $w_b$, and the
initial condition is such that $p_2$ is small, we sometimes have to push the
middle rotor several times back and forth to increase its energy enough to
pass above the ``potential barrier'' created by $U_2$. Conversely, we sometimes
have to brake the middle rotor with some non-trivial controls.
\end{remark}

We now have some piecewise constant control $g(t)$ that can bring the middle
rotor to the final configuration of our choice. It remains to show that we can
make  the external rotors follow some trajectories that have the appropriate
initial and terminal conditions, and such that the force exerted on the middle rotator
closely approximates $g(t)$. We do not prove this in detail,
but we list here the main steps.
\begin{myitem}
\item Since $K^- \leq g(t) \leq K^+$,
it is possible to find piecewise smooth functions $
q_b^*(t)$, $b=1,3$, such that $\sum_{b}w_b(\bar q_2(t)- q_b^*(t)) \equiv  g(t)$, where
$\bar q_2(t)$ is the solution of \eref{eq:controlde3}.
\item Let $\delta > 0$ be small. We can find some smooth trajectories $q_b(t)$ compatible
with the boundary conditions $x^i$ and $x^f$, such that $q_b(t) = q_b^*(t)$ for all 
$t\in [0,T]\setminus A_\delta$, where $A_\delta$ consists of a finite number of
intervals of total length at most $\delta$. We can choose the controls $f_b$ 
so that the $q_b(t)$ constructed here are solutions to \eref{eq:controlde2}
(when $\delta$ is small, $f_b(t)$ is typically very large for $t\in A_\delta$).
\item Since the interaction forces $w_b$ are bounded, their effect during the times
$t\in A_\delta$ is negligible when $\delta$ is small. More precisely, it can
be shown that the solution $q_2(t)$ and $p_2(t)$ of
\eref{eq:controlde2} converge uniformly on $[0, T]$ to 
the solutions $\bar q_2(t)$ and $\bar p_2(t)$ of \eref{eq:controlde3}
when $\delta \to 0$. Therefore, the system is approximately controllable in the
sense of \pref{prop:control}.
\end{myitem}

\section{Numerical illustrations}\label{sec:numerics}

In this section we illustrate some
properties of the invariant measure in the case where $U_i \equiv  0$ and
$W_1= W_3 = -\cos$. 

We use throughout the values
$\gamma_1=\gamma_3=1$ and $\tau_1=0$.
We give examples of how the marginal 
distributions of $p_1, p_2, p_3$ depend on the temperatures $T_1, T_3$ and the
external force $\tau_3$. We apply the numerical algorithm given in
\cite{iacobucci_negative_2011} with time-increment $h
=0.001$. The
resulting graphs are quite independent of $h$. In order to obtain good
statistics and smooth curves,
the probability densities shown below are sampled over $10^8$ units of time
and several hundred bins.

At equilibrium, \ie when $T_1 = T_3 = T$ and $\tau_3=0$ (remember that
$\tau_1 = 0$ in this section), the
marginal law of each $p_i$ has a density proportional to $\exp(-{p_i^2
/2T})$
for $i=1,2,3$. This is obviously not the case out of
equilibrium. Moreover, since we work 
with a finite number of rotors, we do not expect to see any form of local
thermal equilibrium in the bulk of the chain (here the ``bulk'' consists of
only the middle rotor). Clearly, the distribution of
$p_2$ can be quite far from Maxwellian (Gaussian).

In \ref{f:tleftvarie} we show the marginal distributions of $p_1, p_2, p_3$ for
different temperatures and no external force. For each pair of temperatures, we
show the distributions both in linear and logarithmic scale. At equilibrium,
when $T_1 = T_3 = 10$, all three distributions coincide exactly and are
Gaussian. However, when $T_1 \neq 
T_3$, we see that the distribution of $p_2$ is not Gaussian (clearly, the
distribution is not a parabola in logarithmic scale).
\begin{figure}[ht]
\centering
\subfloat[$T_1 = 1$, $T_3=10$]{

\includegraphics[width=1.95in]{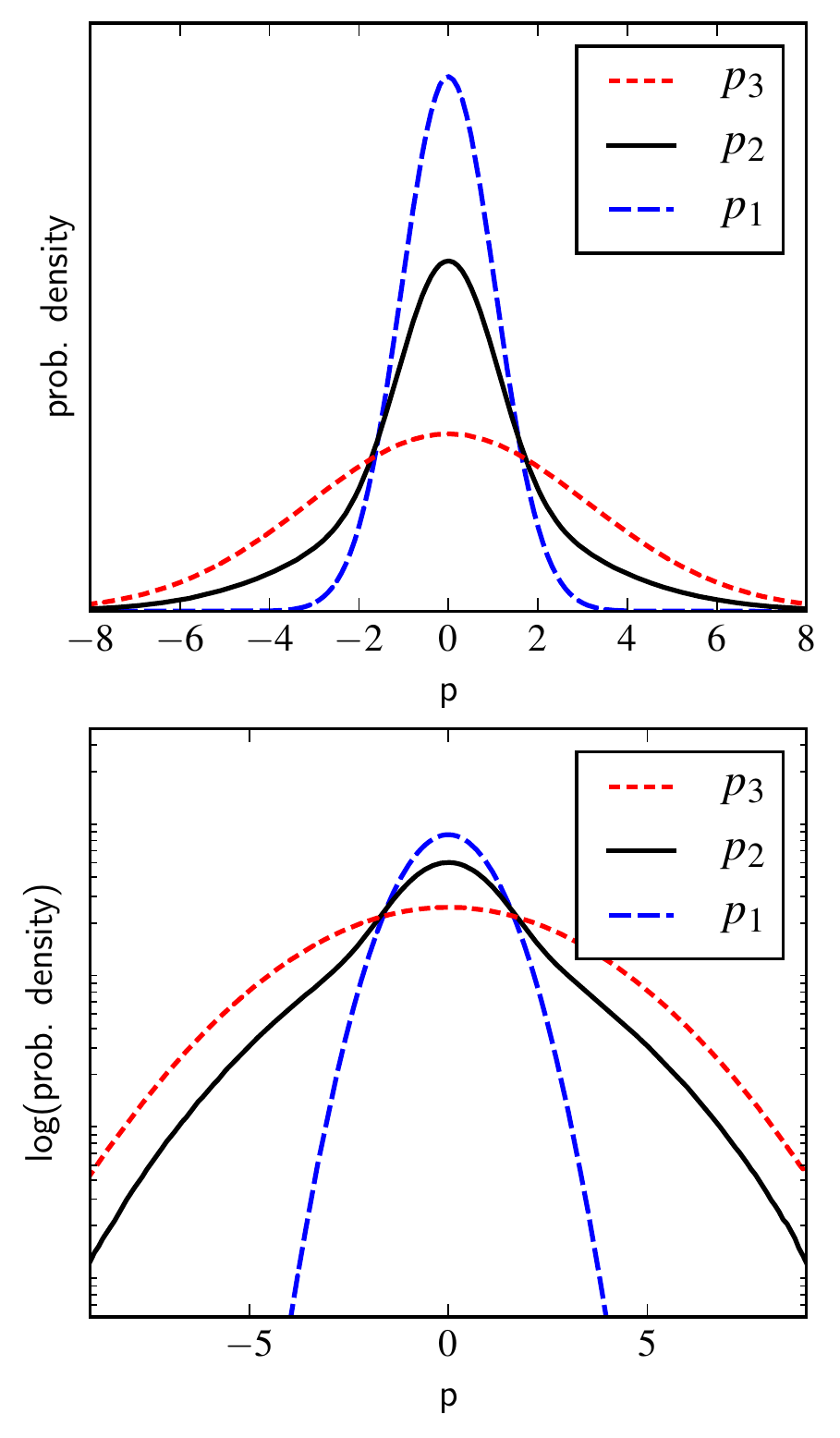}
}
\subfloat[$T_1 = 2$, $T_3=10$]{

\includegraphics[width=1.95in]{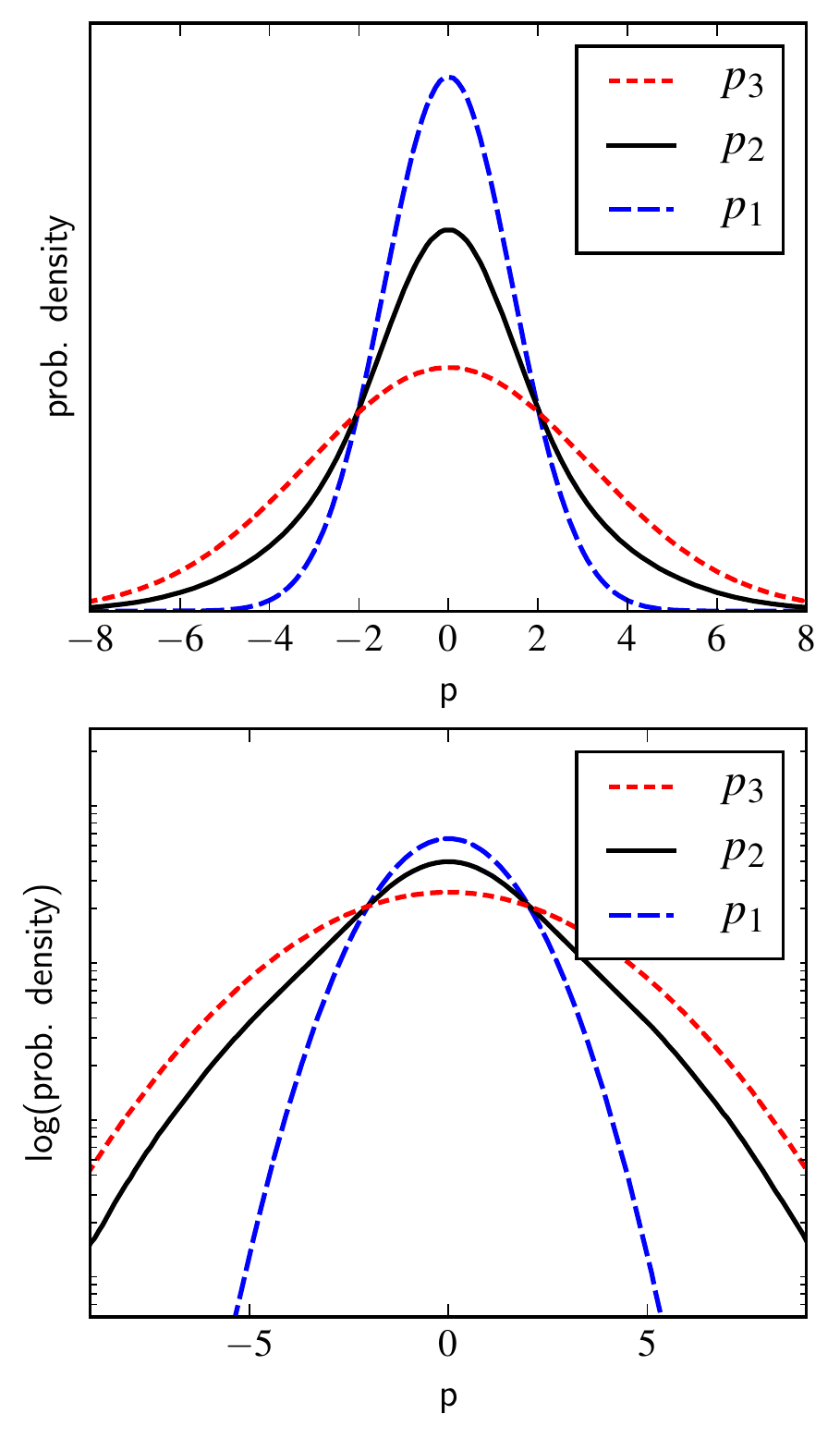}
}
\subfloat[$T_1 = 10$, $T_3=10$]{

\includegraphics[width=1.95in]{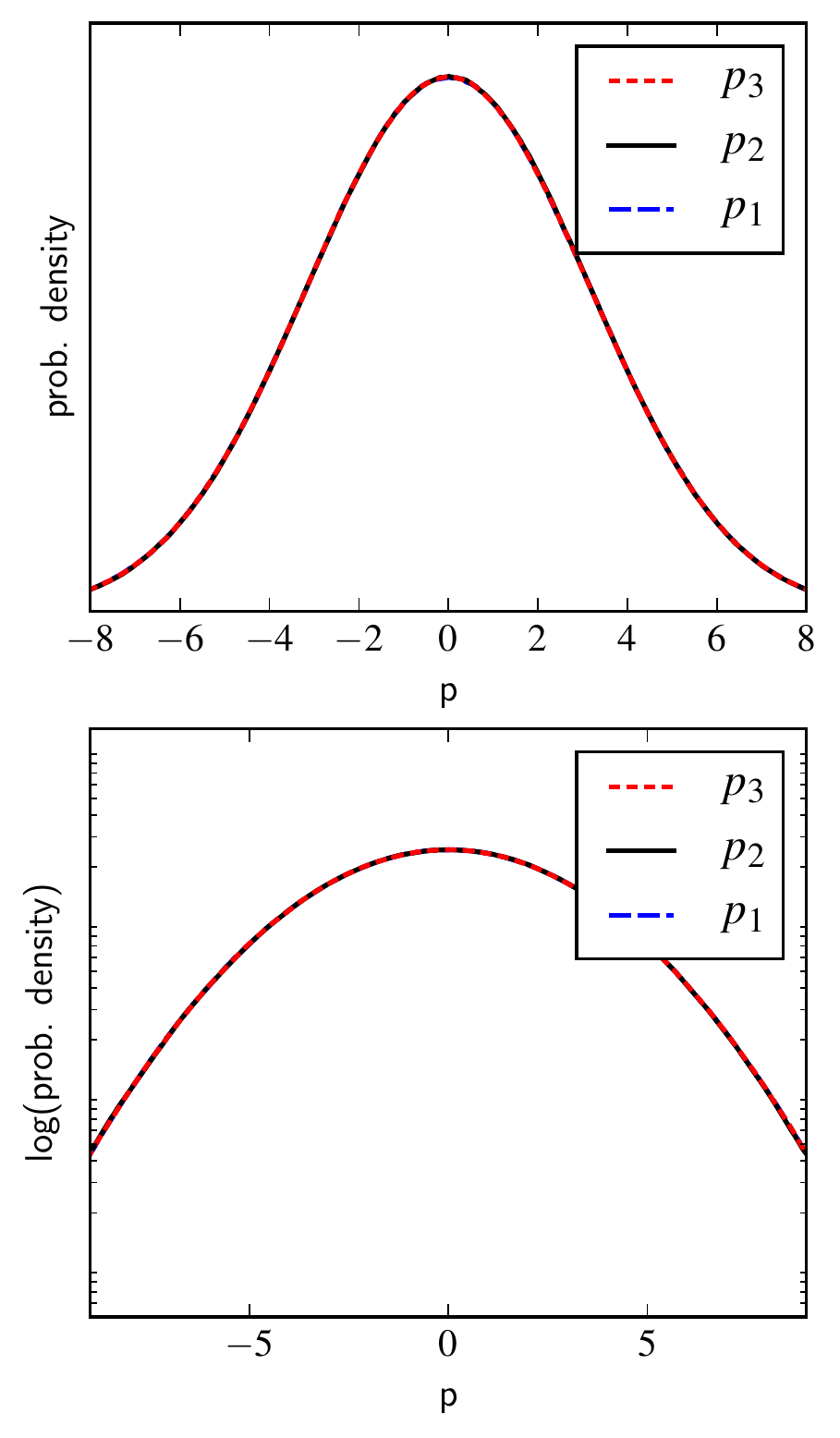}
}\caption{Distribution of $p_1, p_2, p_3$, with no external force and several
temperatures.}\label{f:tleftvarie}
\end{figure}

\begin{figure}[ht]
\centering
\subfloat[$\tau_3 = 0$]{
\includegraphics[width=1.8in]{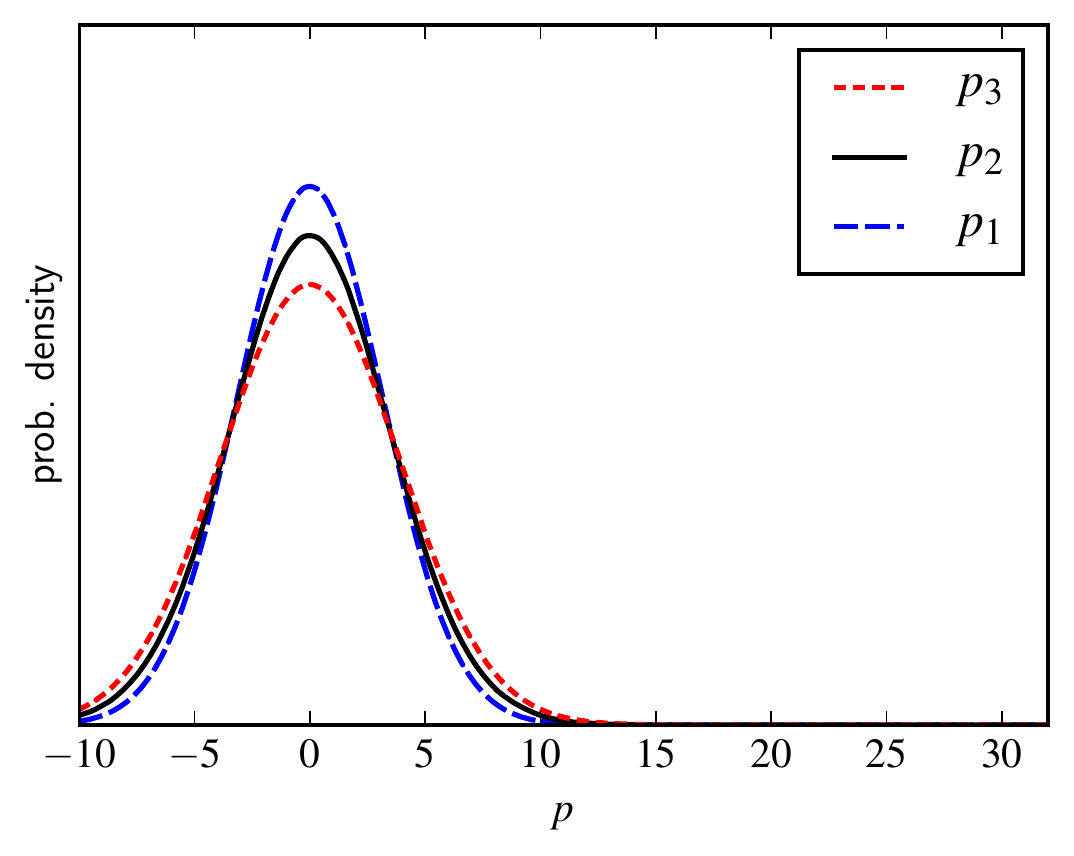}
}
\subfloat[$\tau_3 = 10$]{
\includegraphics[width=1.8in]{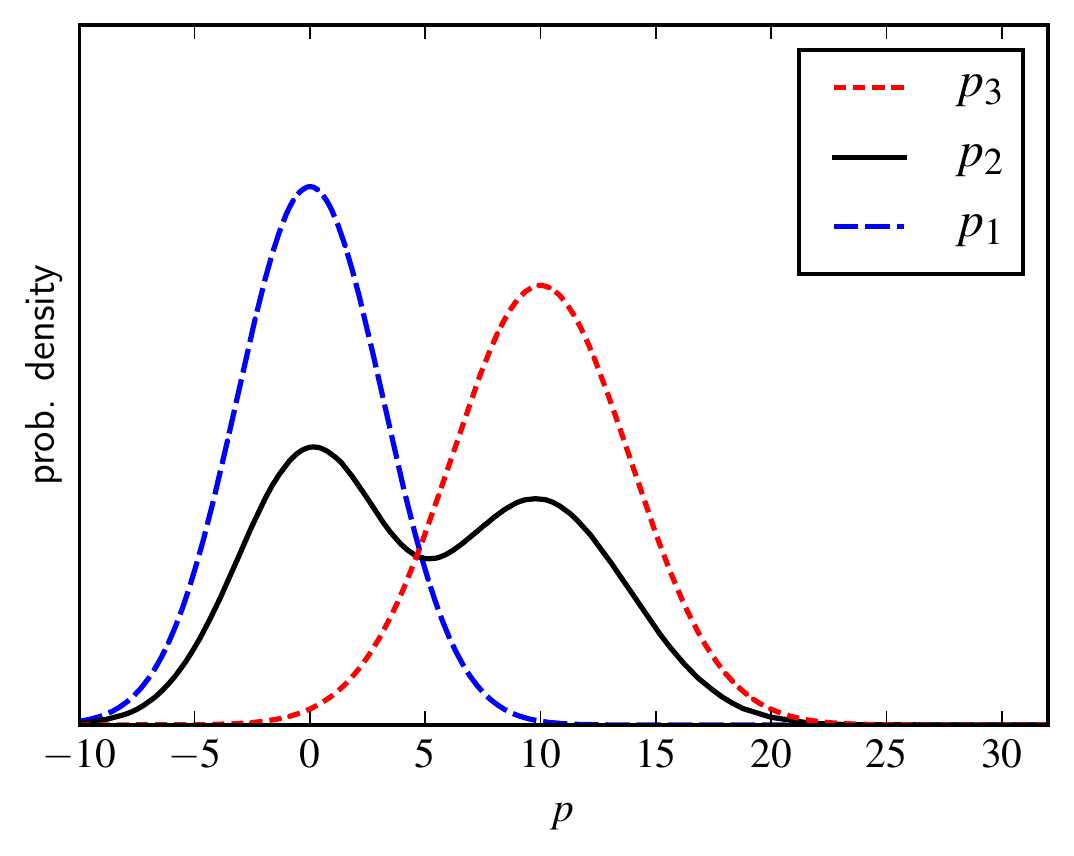}
}
\subfloat[$\tau_3 = 20$]{
\includegraphics[width=1.8in]{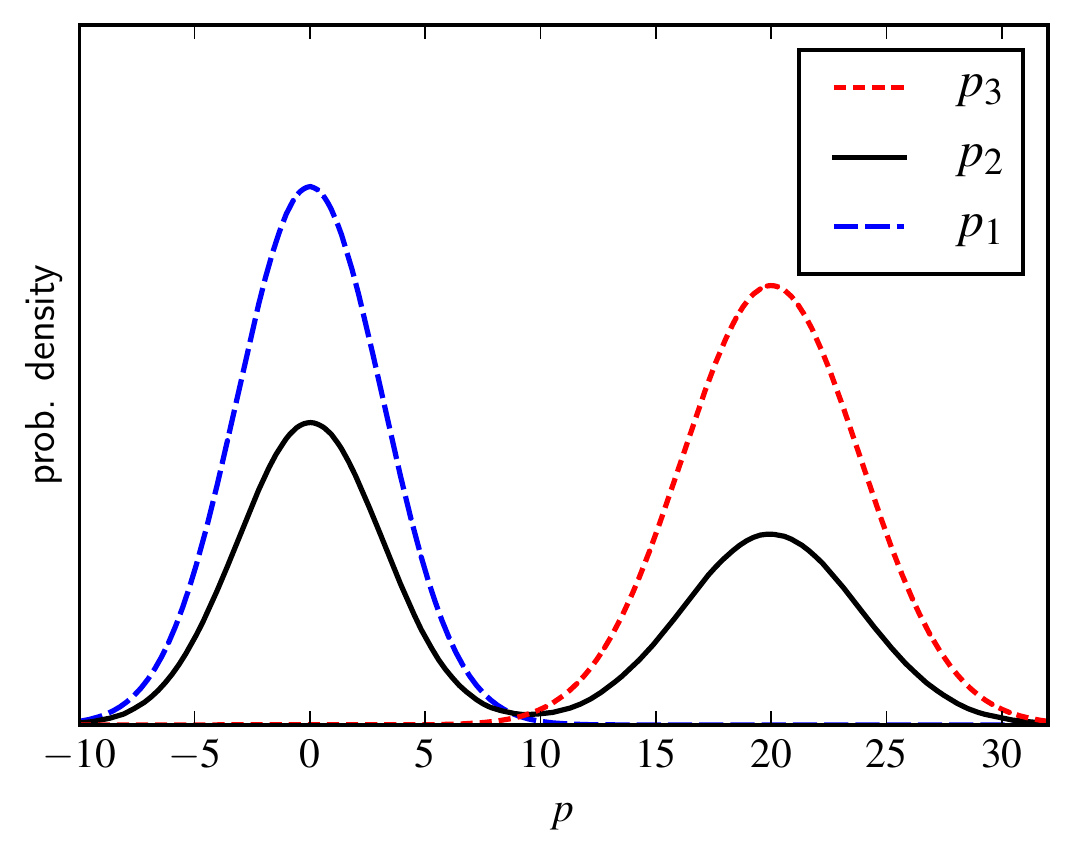}
}
\caption{Distribution of $p_1, p_2, p_3$, with $T_1 = 10, T_3 =
15$ for 3 values of $\tau_3$.}\label{f:tauvarie}
\end{figure}

\begin{figure}[ht]
\centering
\includegraphics[width=6in]{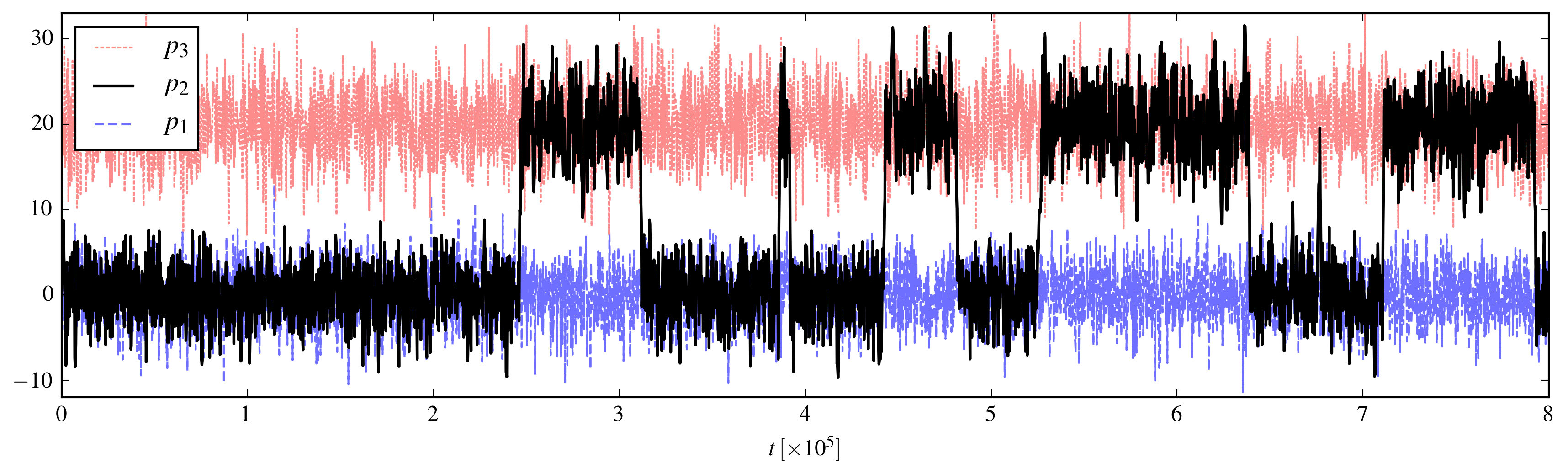}
\caption{Representation of the evolution of $p_1, p_2, p_3$ with $T_1 = 10,
T_3 = 15, \tau_3 = 20$.}\label{f:trajtauvarie}
\end{figure}

We next consider the effect of the external force $\tau_3$ on the marginal
distributions of the $p_i$, for $T_1 = 10$ and $T_3 = 15$. As
illustrated in \fref{f:tauvarie}, the distributions of $p_1$ and $p_3$
are close
to  Gaussians with variance $T_1$ and $T_3$ and mean $0$ and $\tau_3$.
Note that when $\tau_3\neq 0$, the
distribution of $p_2$ has two maxima: one at $0$ and one at
$\tau_3$. The explanation for these two maxima can be found by looking at the
trajectories $p_i(t)$ as shown in \fref{f:trajtauvarie} (for $\tau_3
=20$); $p_1$ fluctuates around $0$, $p_3$ fluctuates
around $\tau_3$, and $p_2$ switches between these two regimes. In the regime
where $p_2$ fluctuates around zero, the rotor 2 interacts strongly with 1 and
weakly with 3 (since then the force $w_3$ oscillates with ``high frequency''
$p_3-p_2 \sim \tau_3$). Inversely, in the regime where $p_2$ fluctuates
around $\tau_3$, it interacts strongly with 3 and only weakly with 1. Other
simulations (not shown here) show that, as expected, the larger $\tau_3$, the
less frequent the switches between these two regimes. The asymmetry of the two
maxima in \fref{f:tauvarie} is explained
by the inequality $T_1 < T_3$, which makes the fluctuations larger in the second
regime, so that the mean sojourn time there is shorter.

\subsection*{Note added in proof}

Based on the results of this paper, an extension to four rotors has been 
obtained more recently \cite{four_rotors_2015}.

\subsection*{Acknowledgments} This work has been partially supported by two
ERC Advanced Grants (``Bridges'' 290843 and ``MALAdY'' 246953). 
C.P. thanks the Paris Dauphine University and the CEREMADE
laboratory, where he held a post-doctoral position during most of this project.
We thank
M.~Hairer, F.~Huveneers, A.~Iacobucci, S.~Olla, and G.~Stolz for valuable advice
and fruitful discussions.

\bibliography{refs}

\end{document}